\documentclass[notitlepage,a4,10.5pt]{article}

\usepackage{graphicx}

\usepackage{amsmath}%
\usepackage{amsfonts}%
\usepackage{amssymb}%
\usepackage{mathrsfs}
\usepackage{amsthm}

\usepackage{authblk}
\usepackage[margin=2cm]{geometry}

\newcommand{\mc}{\mathcal}
\newcommand{\cp}{\times}

\newcommand{\bol}{\boldsymbol}

\newcommand{\abs}[1]{\left\lvert{#1}\right\rvert}
\newcommand{\w}{\wedge}
\newcommand{\lr}[1]{\left({#1}\right)}

\newcommand{\mf}{\mathfrak}
\newcommand{\p}{\partial}

\newcommand{\ti}[1]{\textit{#1}}
\newcommand{\tb}[1]{\textbf{#1}}

\newcommand{\pb}[2]{\left\{ {#1} , {#2} \right\}}
\newcommand{\db}[2]{\left[ {#1} , {#2} \right]}
\newcommand{\mb}[3]{\pb{{#1}}{{#2}}+\db{{#1}}{{#3}}}
\newcommand{\car}{\circlearrowright}

\newtheorem{mydef}{\textit{Def}}
\newtheorem{theorem}{\textit{Theorem}}

\begin{document}
\title{Dissipative Brackets for the Fokker-Planck Equation \\in Hamiltonian Systems and Characterization of Metriplectic Manifolds}
\author[1]{Naoki Sato} 
\affil[1]{Research Institute for Mathematical Sciences, \protect\\ Kyoto University, Kyoto 606-8502, Japan \protect\\ Email: sato@kurims.kyoto-u.ac.jp}
\date{\today}
\setcounter{Maxaffil}{0}
\renewcommand\Affilfont{\itshape\small}

\maketitle

\begin{abstract}
It is shown that the Fokker-Planck equation describing diffusion 
processes in noncanonical Hamiltonian systems exhibits a metriplectic structure, 
i.e. an algebraic bracket formalism that generates the equation   
in consistency with the thermodynamic principles  
of energy conservation and entropy growth. 
First, a microscopic metriplectic bracket is derived for the 
stochastic equations of motion that characterize the random walk of 
the elements constituting the statistical ensemble.
Such bracket is fully determined by the Poisson operator 
generating the Hamiltonian dynamics of an isolated (unperturbed) particle.
Then, the macroscopic metriplectic bracket associated with the evolution
of the distribution function of the ensemble is induced from
the microscopic metriplectic bracket. Similarly, the macroscopic Casimir invariants are inherited from microscopic dynamics.  
The theory is applied to construct the Fokker-Planck equation of an infinite dimensional Hamiltonian system, the Charney-Hasegawa-Mima equation.  
Finally, the canonical form of the symmetric (dissipative) part of the metriplectic bracket 
is identified in terms of a `canonical metric tensor' corresponding to an  Euclidean metric tensor on the symplectic leaves foliated by the Casimir invariants.
\end{abstract}

\section{Introduction}
The purpose of the present paper is to show that, 
given a general Hamiltonian system, 
there exists a bracket formalism that generates
the time evolution of the distribution function 
of the associated statistical ensemble
according to a Fokker-Planck equation that is consistent 
with the first and second
laws of thermodynamics.   

The Hamiltonian formulation of dynamics 
characterizes those systems that preserve energy
throughout the phase space.
The phase space is assigned by the Poisson bracket,
which determines the equations of motion by acting
on the dynamical variable and the Hamiltonian function. 
In general, a Hamiltonian system occurs in noncanonical Hamiltonian form, 
meaning that the Poisson bracket is not expressed
in terms of canonically paired variables \cite{Morrison82, Morrison82_2}.
Hamiltonian systems cannot account for dissipative effects,
i.e. changes that are irreversible from a thermodynamic standpoint.
In a specular way, dissipative dynamics usually 
fails to preserve energy. 
Representative examples are the Navier-Stokes equations in fluid dynamics
and the resistive magnetohydrodynamics equations in plasma physics. 
Furthermore, no shared algebraic structure is known, analogous to the
Poisson algebra of Hamiltonian systems, that
characterizes dissipative dynamics. 
The metriplectic bracket formalism introduced in \cite{Kaufman, Morrison84} aims at reconciling the Hamiltonian formulation of dynamics 
with thermodynamics by introducing a bracket formalism that ensures 
both conservation of energy and non-decrease of entropy, and that reduces to the standard Poisson bracket formalism in the limit of no dissipation.  

In the metriplectic formalism, the time evolution of a variable $F$ is given in the form
\begin{equation}
\frac{d F}{d t}=\lr{F,E,S}=\pb{F}{E}+\db{F}{S}. 
\end{equation}
Here, the metriplectic bracket $\lr{,,}$ is the combination of a Poisson bracket $\pb{}{}$ and a dissipative bracket $\db{}{}$ that have an energy $E$ and an entropy $S$ as generating functions respectively. The dissipative bracket is assumed to be bilinear, non-negative, symmetric, and
to satisfy the Leibniz rule. 
The consistency with the first and second laws of thermodynamics requires that 
(notice that $\pb{E}{E}=0$ due to the alternativity of the Poisson bracket)
\begin{subequations}
\begin{align}
&\frac{dE}{dt}=\db{E}{S}=0,\\
&\frac{dS}{dt}=\pb{S}{E}+\db{S}{S} \geq 0.
\end{align}
\end{subequations}
These conditions are satisfied whenever $\db{E}{S}=\pb{S}{E}=0$, because the dissipative bracket is non-negative. 
One may postulate the stronger requirement that the energy $E$ is a Casimir invariant of the dissipative bracket, and that
the entropy $S$ is a Casimir invariant of the Poisson bracket \cite{Morrison86,Morrison09}: 
\begin{subequations}
\begin{align}
&\db{E}{S}=0~~~~\forall S,\\
&\pb{S}{E}=0~~~~\forall E.
\end{align}
\end{subequations}
If these conditions are satisfied, a single
generating function $\Sigma=S-\beta E$ is sufficient 
to generate the dynamics provided that the 
action of the metriplectic bracket is redefined as
\begin{equation}
\frac{dF}{dt}=\lr{F,\Sigma}=-\beta^{-1}\pb{F}{\Sigma}+\db{F}{\Sigma}.
\end{equation}
Here, the parameter $\beta$, which is needed from dimensional arguments, 
can be related to the inverse temperature of the system. 
The metriplectic formalism has been applied  
in the description of different physical systems with dissipation, 
such as convection-diffusion equations 
and magnetohydrodynamics equations with viscosity and resistivity (see for example \cite{Grmela,Materassi12,Coquinot}). 

While the metriplecitc structure described above is fully compatible with thermodynamics,
it still lacks information on the nature of the dissipative part of the bracket.  
This fact can be seen explicitly if one considers the expression of the metriplectic bracket in finite dimensions:
\begin{equation}
\dot{F}=F_i\mc{J}^{ij}E_j+F_i g^{ij}S_j.
\end{equation}
Here, the upper dot stands for time derivative, lower indexes denote partial derivatives, e.g. $E_j=\p E/\p x^j$, where $x^j$
is the $j$th coordinate of an $n$ dimensional coordinate system $\bol{x}=\lr{x^1,...,x^n}$, and 
the standard convention of summation over repeated indexes is used. 
The antisymmetric contravariant tensor $\mc{J}^{ij}=-\mc{J}^{ji}$, $i,j=1,...,n$, 
is a Poisson operator, which is mathematically characterized 
by the axioms that define a Poisson algebra.   
However, the geometrical nature and physical origin of the symmetric non-negative contravariant tensor $g^{ij}$
associated with dissipation is unclear. 
In other words, it is desirable to determine whether the tensor $g^{ij}$ can be derived from first principles,
and if any relationship exists between $\mc{J}^{ij}$ and $g^{ij}$.
We will see that this is indeed the case for the Fokker-Planck construction discussed in the present study.

Several authors have proposed the following form for the tensor $g^{ij}$:
\begin{equation}
g^{ij}=\mc{J}^{ik}\mc{J}^{jk}.\label{gijd}
\end{equation}
Here, it is tacitly assumed that the right-hand side contains a summation over the index $k$,  
while the lowering of the index due to the contraction with an Euclidean metric tensor is omitted.
This same convention will be used throughout the paper.
The dissipative bracket associated with the tensor \eqref{gijd} is usually referred to as a `double bracket', 
because it corresponds to the bracket generated by minus the square of the Poisson operator \cite{Morrison09}.  
The advantage of the form \eqref{gijd} is that the dissipative bracket retains the Casimir invariants
of the Poisson operator, while the rate of change in the energy has a semi-definite sign, 
$\dot{E}=-\beta\db{E}{E}\leq 0$, provided that the generating function is $\Sigma=S-\beta E$ with $S$ a Casimir invariant of the Poisson bracket. This property can be used to determine the equilibria of a Hamiltonian system
with given values of the Casimir invariants. 
The procedure is discussed for the ideal Euler equations in \cite{Vallis}, 
and it has been applied to calculate ideal magnetohydrodynamics equilibria in \cite{Furukawa}.  
Examples of repeated application of two Poisson brackets to describe dissipative behavior 
in the context of kinetic theory and gradient flow can be found in \cite{Holm} and \cite{Bloch}.

In this study, we start with a general Hamiltonian system,
and consider many identical copies of such system to define a statistical ensemble.
The ensemble may consist of identical particles, as in an ideal gas,
or a set of identical equations (infinite dimensional
Hamiltonian systems), e.g. an ensemble of magnetohydrodynamics systems. 
Then, the interaction among the elements of the ensemble is modeled 
in terms of random fluctuations in the energy and a 
dissipative force that sets in due to the constraint imposed by the
conservation of total energy. This procedure is the Hamiltonian counterpart  
of the construction of the Langevin equation.  
The result is a set of stochastic equations of motion 
that inherit the geometrical structure of the original Poisson operator.
These equations are then translated into the corresponding Fokker-Planck equation according to the standard procedure \cite{Gardiner,Risken}. 
The expression of the Fokker-Planck equation was derived in \cite{Sato18,Sato19}
and the analysis of the corresponding diffusion operator is discussed in \cite{Sato19_2}. Our task is to show that this equation can be 
written in terms of a metriplectic bracket. 

There are two aspects of this construction that need to be stressed. 
First, the statistical behavior of the ensemble is modeled around the
expectation that thermodynamic equilibrium should be in the form of
a generalized Boltzmann state. In other words, if $f$ denotes the distribution function, the solution of the Fokker-Planck equation
in the limit $t\rightarrow\infty$ should satisfy
\begin{equation}
f\propto J\exp\left\{-\beta H-\mu_k C^k\right\},\label{feq0}
\end{equation}
where $H$ denotes the Hamiltonian function (the energy of a single constituent of the ensemble), $C^k$ the $k$th Casimir invariant of the original Poisson operator, $\mu_k$ the corresponding Lagrange multiplier (chemical potential),
and $J$ the Jacobian determinant of the transformation of variables linking the
coordinate system used to define the distribution function to
the invariant measure assigned by Liouville's theorem. 
Such preserved volume element is always available (at least locally)
in finite dimensional Hamiltonian systems due to the Lie-Darboux theorem,
which assigns a local phase space measure consisting of canonically conjugated variables plus a given number of Casimir invariants (see \cite{Arnold,deLeon,Littlejohn82}).  
Distributions in the form \eqref{feq0} have been proposed in \cite{Yoshida14}
and reflect the fact that, since for a noncanonical Hamiltonian system $J\neq 1$ in general, equiprobability of microstates cannot be enforced in arbitrary coordinate systems. In particular, the knowledge of $J$ is necessary to
properly identify the entropy of the system, which is not covariant \cite{Jaynes,Sato16}. We will provide an example on how to calculate $J$
when the base Hamiltonian system is infinite dimensional and the notion
of invariant measure is non-trivial. 
This is achieved by expanding the solution in a discrete
basis of the relevant function space, and by using the coefficients
of the expansion as dynamical variables.
Examples of this approach in the context of magnetohydrodynamics
can be found in \cite{Ito,Kraichnan}.
In our example, we will consider the Charney-Hasegawa-Mima equation
\cite{Hasegawa78,Weinstein,Tassi,Swaters}, derive the
preserved phase space volume in Fourier space, and then
construct the Fokker-Planck equation in metriplectic form. 

The second remark concerns the form of the tensor $g^{ij}$ 
appearing in the Fokker-Planck equation. 
We find that $g^{ij}$ appears in the form \eqref{gijd}
inside the dissipative bracket generating the diffusion operator
of the Fokker-Planck equation provided that the 
the coordinate system is chosen to be such that $J=1$, i.e.
\begin{equation}
\db{F}{G}=\frac{D}{2}\int_{\Omega}f\frac{\p}{\p x^i}\lr{\frac{\delta F}{\delta f}}_{\beta}\mc{J}^{ik}\mc{J}^{jk}\frac{\p}{\p x^j}\lr{\frac{\delta G}{\delta f}}_{\beta}dV.
\end{equation}
The details on the notation used here will be given in the relevant sections.
This result shows that the dissipative part of the metriplectic bracket
is related to the Poisson operator.
Furthermore, if we consider the simplest 2 dimensional setting, 
it follows that the canonical form for the covariant version $\mf{g}$ of the tensor $g$ is given by an Euclidean metric tensor
in the phase space coordinates $\lr{p,q}$, i.e. 
\begin{equation}
\mf{g}= dp\otimes dp+dq\otimes dq.
\end{equation}
This result should be compared with the canonical form of the symplectic 2 form
associated with Hamiltonian mechanics, $\omega=dp\w dq=dp\otimes dq-dq\otimes dp$. 

The present paper is organized as follows. 
In section 2 we derive the `microscopic' metriplectic bracket
that generates the stochastic dynamics of each element of the statistical ensemble.	
In section 3 we induce the corresponding `macroscopic' metriplectic bracket
that generates the Fokker-Planck equation for the time evolution of the
distribution function.  
In section 4 we construct the metriplectic bracket
for the Fokker-Planck equation of an ensemble of periodic Charney-Hasegawa-Mima
equations.
In section 5 we discuss certain geometric aspects that characterize
the metriplectic bracket obtained in the present study.
Concluding remarks are given in section 6. 
	


\section{Dissipative Brackets for Stochastic Dynamics}
Aim of the present section is to obtain the dissipative
brackets associated with the stochastic dynamics (random walk) 
of diffusion processes in Hamiltonian systems.
It will be shown that these brackets exhibit a metriplectic structure. 

\subsection{Poisson brackets for microscopic dynamics}

Let $\Omega\subset\mathbb{R}^n$ denote a smoothly bounded domain with boundary $\p\Omega$ 
and $\bol{x}=\lr{x^1,...,x^n}$ a coordinate system in $\Omega$ with tangent basis $\bol{\p}=\lr{\p_1,...,\p_n}$.  
We consider the motion of an ensemble of $N$ particles in $\Omega$. 
When isolated from the others, the trajectory of each particle 
evolves according to the noncanonical Hamiltonian system
\begin{equation}
\dot{x}^i=\mc{J}^{ij}H_j,~~~~i=1,...,n.\label{xdot}
\end{equation} 
Here, lower indexes are used for partial derivatives, e.g. $H_j=\p H/\p x^j$, 
the energy $H\in C^{\infty}\lr{\Omega}$ is the Hamiltonian function, 
and $\mc{J}\in\bigwedge^2 T\Omega$ a Poisson operator, i.e. a 
bivector field (antisymmetric matrix) satisfying the Jacobi identity
\begin{equation}
\mc{J}^{im}\mc{J}_m^{jk}+\mc{J}^{jm}\mc{J}_m^{ki}+\mc{J}^{km}\mc{J}_m^{ij}=0,~~~~i,j,k=1,...,n.\label{JI_m}
\end{equation}  
In the following, we assume that $\mc{J}^{ij}\in C^{\infty}\lr{\Omega}$, $i,j=1,...,n$. 
Both \eqref{xdot} and \eqref{JI_m} can be expressed in terms of the Poisson bracket
\begin{equation}
\left\{\alpha,\beta\right\}_{\rm m}=\alpha_i\mc{J}^{ij}\beta_j,~~~~\alpha,\beta\in C^{\infty}\lr{\Omega}.
\end{equation}
The lower index ${\rm m}$, which stands for `microscopic', 
is used to distinguish this Poisson bracket from the `macroscopic'
Poisson bracket associated with the dynamics of the entire ensemble.  
This second bracket will be derived in the next section.
Then, equations \eqref{xdot} and \eqref{JI_m} read as
\begin{equation}
\dot{x}^i=\left\{x^i,H\right\}_{\rm m},~~~~i=1,...,n,
\end{equation}
and
\begin{equation}
\pb{\alpha}{\pb{\beta}{\gamma}_{\rm m}}_{\rm m}+\car=
\pb{\alpha}{\pb{\beta}{\gamma}_{\rm m}}_{\rm m}+
\pb{\beta}{\pb{\gamma}{\alpha}_{\rm m}}_{\rm m}+
\pb{\gamma}{\pb{\alpha}{\beta}_{\rm m}}_{\rm m}
=0.
\end{equation}
In this notation, the symbol $\car$ indicates summation of even permutations.

It is useful to spend some words on the conservation laws, relevant
for the present study, that
arise from the properties of the Poisson operator $\mc{J}$.
First, observe that antisymmetry $\mc{J}^{ij}=-\mc{J}^{ji}$ 
determines conservation of energy:
\begin{equation}
\dot{H}=H_i\mc{J}^{ij}H_{j}=\frac{1}{2}\lr{H_i\mc{J}^{ij}H_j+H_j\mc{J}^{ji}H_i}=0.
\end{equation}
Similarly, the Jacobi identity \eqref{JI_m} is associated with a conservation law. This fact can be seen explicitly when $\mc{J}$ is invertible. The inverse $\omega^{ij}=-\omega^{ji}$ is given by the components of a 2 form $\omega\in\bigwedge^2 T^\ast\Omega$, called the symplectic 2 form. We have
 \begin{equation}
d\omega
=\sum_{i<j<k}\left(\frac{\partial\omega_{ij}}{\partial x^{k}}+\frac{\partial\omega_{jk}}{\partial x^{i}}+\frac{\partial\omega_{ki}}{\partial x^{j}}\right) dx^{i}\wedge dx^{j}\wedge dx^{k}.
\end{equation}
Multiplying each component of this $3$ form by $\mathcal{J}^{li}\mathcal{J}^{mj}\mathcal{J}^{nk}$ and summing over $i,j,k$, we obtain:
\begin{equation}
\begin{split}
\mathcal{J}^{li}\mathcal{J}^{mj}\mathcal{J}^{nk}\left(\frac{\partial\omega_{ij}}{\partial x^{k}}+\frac{\partial\omega_{jk}}{\partial x^{i}}+\frac{\partial\omega_{ki}}{\partial x^{j}}\right)
=&-\mc{J}^{mj}\mc{J}^{nk}\omega_{ij}\mc{J}^{li}_k
-\mc{J}^{li}\mc{J}^{nk}\omega_{jk}\mc{J}^{mj}_i
-\mc{J}^{li}\mc{J}^{mj}\omega_{ki}\mc{J}^{nk}_j\\
=&\delta_{i}^{m}\mc{J}^{nk}\mc{J}^{li}_k+
\delta^{n}_{j}\mc{J}^{li}\mc{J}^{mj}_i+
\delta^{l}_{k}\mc{J}^{mj}\mc{J}^{nk}_j\\
=&\mc{J}^{nk}\mc{J}^{lm}_k+
\mc{J}^{li}\mc{J}^{mn}_i+
\mc{J}^{mj}\mc{J}^{nl}_j.
\end{split}
\end{equation}
Thus, the closure $d\omega=0$ of the symplectic 2 form $\omega$ is equivalent to the Jacobi identity. 
On the other hand, the equations of motion \eqref{xdot} take the form
\begin{equation}
i_{\dot{\bol{x}}}\omega=-dH,
\end{equation}
which, together with $d\omega=0$, imply the conservation of $\omega$ along the flow generated by $\dot{\bol{x}}$:
\begin{equation}
\mf{L}_{\dot{\bol{x}}}\omega=\lr{di_{\dot{\bol{x}}}+i_{\dot{\bol{x}}}d}\omega=-ddH=0.
\end{equation}
In this notation, $i$ is the contraction operator and $\mf{L}$
the Lie derivative. 

When $\mc{J}$ is not invertible, a similar result applies. 
More precisely, the Lie-Darboux theorem \cite{Arnold,deLeon,Littlejohn82} ensures that in every sufficiently small 
neighborhood $U\subset\Omega$ where the rank $2r=n-m$ of $\mc{J}$ is constant 
there exist $2r$ local coordinates $\bol{z}=\lr{p^1,...,p^r,q^1,...,q^r}$ and $m=n-2r$ Casimir invariants $\lr{C^1,...,C^m}$, with the property that
\begin{equation}
\mc{J}=\sum_{i=1}^r \p_{q^i}\w \p_{p^i},~~~~\omega=\sum_{i=1}^{r}dp^i\w dq^i,~~~~i_{\dot{\bol{z}}}\omega=-dH,\label{eq20}
\end{equation}
and
\begin{equation}
\mc{J}^{ij}C^{k}_{j}=0,~~~~i=1,...,n,~~k=1,...,m.\label{eq21}
\end{equation}
For completeness, a proof of the Lie-Darboux theorem in the case of a degenerate 2 form $\omega$ is given in appendix A. 
Notice that the equation \eqref{eq20} ensures the conservation of the 2 form $\omega$  with respect to $\dot{\bol{z}}$, while equation \eqref{eq21} implies that the Casimir invariants $C^{k}$ are constants of motion for any choice of the Hamiltonian function. 
The conservation of the Casimir invariants is expressed through the Poisson bracket as
\begin{equation}
\dot{C}^k=\pb{C^k}{H}_{\rm m}=0~~~~\forall H,~~k=1,...,m.
\end{equation}

Finally, there is a third conservation law, Liouville's theorem. 
This theorem states that the local phase space measure
provided by the Lie-Darboux theorem 
\begin{equation}
d\Pi=dp^1\w ...\w dp^r \w dq^1\w ... \w dq^r \w dC^1\w ... \w dC^m,\label{dPi} 
\end{equation}
is conserved by the Hamiltonian flow, i.e. 
\begin{equation}
\mf{L}_{\dot{\bol{x}}}d\Pi=\lr{\p_{p^i}\dot{p}^i+\p_{q^i}\dot{q}^i}d\Pi=0~~~~\forall H.\label{Liouville}
\end{equation}
Liouville's theorem can be easily verified by
recalling that $\dot{p}^i=-H_{q^i}$, $\dot{q}^i=H_{p^i}$, and $\dot{C}^k=0$, $i=1,...,r$, $k=1,...,m$.
The invariant measure $d\Pi$, which is independent of the choice of $H$,
is at the core of the classical formulation of statistical mechanics.

\subsection{Dissipative brackets for random walks in noncanonical Hamiltonian systems}

Dissipation occurs when particles interact with each other (in the following, the word particle is used to specify an element of the ensemble).  
The interaction causes the energy $H$ of each particle to fluctuate, 
resulting in a random walk (diffusion process) that drive the ensemble  
toward thermodynamic equilibrium. 
When a scattered particle gains energy, 
an effective friction (dissipative) force sets in. 
These competitive processes are bound together 
by the constraint imposed by the conservation of the total energy of the system.  
Denoting with $\delta H$ the energy fluctuation and with $\mc{F}_j$ the $j$th component of the friction force, the equations of motion take the form
\begin{equation}
\dot{X}^i=\mc{J}^{ij}\lr{H_j+\delta H_j-\mc{F}_j},~~~~i=1,...,n.\label{Xi}
\end{equation}
Here, the uppercase letter $X^i$ is used to distinguish  
the trajectory $X^i\lr{t}$ from the unperturbed trajectory $x^i\lr{t}$
resulting from integration of \eqref{xdot}. 
In order to characterize $\delta H$ and $\mc{F}_j$ some physical assumptions are needed on the relaxation process under consideration.
Below, we follow the construction 
of \cite{Sato18} and \cite{Sato19} to obtain $\delta H$ and $\mc{F}_j$,
and then derive the dissipative bracket associated with \eqref{Xi}.

Let $f\lr{\bol{x},t}$ be the particle distribution function 
(probability density function) in the coordinate system $\lr{x^1,...,x^n}$.
If the system is thermodynamically isolated, the equilibrium 
distribution function should be determined by two factors:
the constraints (conserved quantities) 
and the intrinsic geometry of the forces associated with the relaxation process. Hence, we expect that
\begin{equation}
\lim_{t\rightarrow\infty}f=\frac{1}{Z}J\exp\left\{-\beta H-\mu_k C^k\right\}.\label{feq}
\end{equation} 
In this expression, $Z\in\mathbb{R}$ is a normalization constant  
associated with conservation of total particle number, 
i.e. the constraint 
\begin{equation}
N=\int_{\Omega}f\,dV=1. 
\end{equation}
Here, $dV=dx^1\w ... \w dx^n$ is the volume element. We have
\begin{equation}
Z=\int_{\Omega}J\exp\left\{-\beta H-\mu_kC^k\right\}dV.
\end{equation}
Similarly, the quantities $\beta,\mu_1,...,\mu_k\in\mathbb{R}$ are
physical parameters that can be interpreted as the Lagrange multipliers
of a variational principle where the entropy measure  
\begin{equation}
S=-\int_{\Omega}f\log\lr{\frac{f}{J}}dV,\label{S}
\end{equation}
is extremized under the constraints
\begin{equation}
E=\int_{\Omega}fH\,dV,~~~~\mf{C}^k=\int_{\Omega}fC^k\,dV,~~~~k=1,...,m,
\end{equation}
representing conservation of total energy and total Casimir invariants
during the relaxation process. 
We will refer to $\beta$ as the inverse temperature. 
The equilibrium
distribution function \eqref{feq} is thus obtained from the variational principle
\begin{equation}
\delta\lr{S-\alpha N-\beta E-\mu_k\mf{C}^k}=0
\end{equation}
with $\alpha=\log{Z}-1$ and where the variation is carried out with respect to $f$. The breaking of a constraint will then be represented by the vanishing of the
corresponding Lagrange multiplier. 

The remaining nontrivial quantity is the function $J$ appearing
on the right-hand side of \eqref{feq}, which is related
to the definition of the entropy measure $S$. 
The geometric factor $J$ is a manifestation of the fact that
the `homogenization' of a density, 
such as the broadening of the particle distribution function
accompanied by entropy growth, is coordinate dependent:  
a flat density profile $f_y$ in a given reference frame $\lr{y^1,...,y^n}$
may correspond to an heterogeneous distribution $f=f_yJ$ in a different
coordinate system $\lr{x^1,...,x^n}$ due to the inhomogeneous Jacobian weight $J$ of the transformation. Indeed, 
the distributions $f_y$ and $f$ are related by  
\begin{equation}
f_ydV_y=f_yJdV=fdV,
\end{equation}
where $dV_y=dy^1\w ... \w dy^n=Jdx^1\w ...\w dx^n=JdV$ is the volume element.  
These considerations can be summarized by noting that 
Shannon's information entropy measure $S\left[f\right]=-\int_{\Omega}f\log{f}dV$ for a distribution function $f$ is non-covariant \cite{Sato16}, i.e.  
\begin{equation}
S\left[f_y\right]=-\int_{\Omega}f_y\log{f_y}\,dV_y=S\left[f\right]+\int_{\Omega}f\log{J}\,dV.
\end{equation}
It is now clear that, in order to fully characterize the
equilibrium state \eqref{feq}, one needs to determine $J$,
which encapsulates the intrinsic geometric properties of the physical forces. 
This amounts at finding a coordinate system 
(provided that such a coordinate system exists)
where the relaxation process
is `homogeneous', in the sense that the equilibrium state 
only depends on the constraints affecting the system.
This could accomplished by establishing the conditions 
under which the interaction force $\delta H_j$ is suitably represented
by a spatially homogeneous random process (typically, a Brownian motion) 
that enables the derivation of the effective collision operator of the system.

Unfortunately, the cost of this task usually overcomes the benefit of the theory. 
Therefore, one is led to postulate an ergodic ansatz \cite{Moore} for
the perturbed dynamics,
namely that the transformation $T:\Omega\rightarrow\Omega$ generated 
by the flow $\dot{X}^i$ is ergodic in a submanifold $\Omega_{C}=\left\{\bol{x}\in\mathbb{R}^n~\rvert~C^1=c^1,...,C^{m}=c^m\right\}$ with $c^i\in\mathbb{R}$, $i=0,...,m$. 
This means that the particle eventually explores the whole 
reduced phase space $\Omega_{C}$, and, in the limit $t\rightarrow\infty$ and for a given value of the energy,  
the time spent by a particle 
in a certain region of $\Omega_{C}$ 
is proportional to the volume of that region. 
The reason why the constraints $C^k$, $k=1,...,m$, are not broken by the relaxation process 
is that energy fluctuations always result in
scatterings that are tangential to the Casimir isosurfaces, $\dot{C}^k=C^k_i\mc{J}^{ij}\lr{H_j+\delta H_j-\mc{F}_j}=0$. 
Violation of the Casimir invariants occurs when the structure of the phase space itself is subject to fluctuations,
implying that the effective Poisson operator $\mc{J}'$ contains a deviation term, $\mc{J}'=\mc{J}+\delta\mc{J}$.
However, this second scenario is not examined in the present study. 

The essential prerequisite for the ergodic hypothesis to hold is that
the transformation $T$ is measure preserving, i.e. that there exists
and invariant measure $JdV$ such that
\begin{equation}
\mf{L}_{\dot{\bol{X}}}JdV=\p_i\lr{J\dot{X}^i}\,dV=0.
\end{equation}
Since the form of $\dot{X}^i$ is not known a priori, the condition above is replaced by the following
requirement on the unpertubed single particle motion:
\begin{equation}
\mf{L}_{\dot{\bol{x}}}JdV=\p_i\lr{J\dot{x}^i}\,dV=\p_i\lr{J\mc{J}^{ij}}H_{j}\,dV=0~~~~\forall H,
\end{equation}
which implies
\begin{equation}
\p_i\lr{J\mc{J}^{ij}}=0,~~~~j=1,...,n.\label{IM}
\end{equation}
When a nontrivial solution exists, this condition assigns an invariant measure $JdV$ for any choice of the Hamiltonian $H$, 
so that fluctuations in $H$ do not affect the Jacobian weight $J$.
As a consequence of the Lie-Darboux and Liouville's theorems, equation \eqref{IM} 
always has a local solution in Hamiltonian systems. 
Then, $JdV=d\Pi$, with $d\Pi$ the phase space measure \eqref{dPi}. 
It follows that the thermodynamically consistent entropy measure is given by 
Shannon's information entropy measure for the distribution function $f_{y}=f/J$ on the invariant measure $dV_y=JdV$:
\begin{equation}
S=S\left[f_y\right]=-\int_{\Omega}f_y\log{f_y}\,dV_y=-\int_{\Omega}f\log\lr{\frac{f}{J}}\,dV,
\end{equation}
which is the expression postulated in \eqref{S} but with the difference that
now $J$ is known.

In the following, we simplify the notation by assuming that the coordinate system $\lr{x^1,...,x^n}$
already spans the invariant measure, implying that $J=1$ and $f_y=f$. 
The ergodic hypothesis then enables the interchange of
ensemble averages with time averages. 
Here, we assume that the fluctuating force $\delta H_j$ has vanishing ensemble average 
(any term with non-vanishing ensemble average should be reabsorbed in $H_j$)
and replace it with a spatially homogeneous random
process with zero time average,
\begin{equation}
\delta H_j=D^{1/2}\Gamma_j,~~~~j=1,...,n.\label{dHj}
\end{equation}
In the equation above, $D$ is a non-negative real constant (diffusion parameter) representing the amplitude of the 
fluctuations and $\Gamma_j$ the $j$th component of an $n$ dimensional Gaussian white noise process. 
Integrating equation \eqref{dHj} leads to the following expression for the energy fluctuation
\begin{equation}
\delta H=D^{1/2}x^j\Gamma_j.
\end{equation} 
Finally, we assume that the friction force acts against the unperturbed particle velocity
through a proportionality coefficient (friction coefficient) $\gamma$,
\begin{equation}
\mc{F}_j=-\gamma\mc{J}^{jk}H_k.\label{Fj}
\end{equation}
That this is the correct expression for $\mc{F}_j$ can be verified
by showing that the resulting stochastic dynamics produces the desired equilibrium state, equation \eqref{feq}.
This is proved in \cite{Sato18} and \cite{Sato19}. We will review this fact in the next section when discussing
the bracket formalism for the associated Fokker-Planck equation.
We will also see that the value of the spatial constant $\gamma$ 
is related to the diffusion parameter $D$ 
by the constraint imposed by conservation of energy
through the inverse temperature $\beta$.
This fluctuation-dissipation relation effectively determines
the temperature of the system at each time $t$.

The equation of motion \eqref{Xi} expressed in the coordinates spanning the invariant measure now takes the form
\begin{equation}
\dot{X}^i=\mc{J}^{ij}\lr{H_j+D^{1/2}\Gamma_j+\gamma\mc{J}^{jk}H_k}.\label{Xi2}
\end{equation} 
In this expression there is a violation of the summation convention since
the $j$ index appears always as an upper index in the third term on the right-hand side. 
This is because, in a general setting, $\gamma$ should not be interpredted simply as a spatial constant,
but rather as a covariant tensor $\gamma_{jk}$. 
Then, the friction force reads as
\begin{equation}
\mc{F}_j=-\gamma_{jk}\mc{J}^{kl}H_l.
\end{equation} 
We postulate that $\gamma_{jk}=\gamma\delta_{jk}$
in the coordinates spanning the invariant measure, which
gives \eqref{Fj}. As already mentioned, the consistency of such hypothesis
with energy conservation and entropy law 
will be verified a posteriori. 

It is useful to write equation \eqref{Xi2} when $\mc{J}$
is the 2 dimensional symplectic matrix
\begin{equation}
\mc{J}=\begin{bmatrix}0&-1\\1&0\end{bmatrix},
\end{equation}
and $\lr{x^1,x^2}=\lr{p,q}$ are canonical coordinates. 
Denoting stochastic variables with uppercase letters, we have
\begin{subequations}
\begin{align}
\dot{P}&=-H_q-D^{1/2}\Gamma_q-\gamma H_p,\\
\dot{Q}&=H_p+D^{1/2}\Gamma_p-\gamma H_q,
\end{align}
\end{subequations}
which is the phase space version of the classical Langevin equation
\begin{equation}
m\ddot{\bol{X}}=\bol{F}+D^{1/2}\bol{\Gamma}-\gamma\dot{\bol{X}}.
\end{equation}
Here, $m$ is the particle mass and $\bol{F}$ represents force. 

Given $\alpha,\beta\in C^{\infty}\lr{\Omega}$, the metriplectic bracket 
associated with equation \eqref{Xi2} can be identified to be
the combination of the following brackets:
\begin{subequations}\label{mmb}
\begin{align}
\pb{\alpha}{\beta}_{\rm m}&=\alpha_i\mc{J}^{ij}\beta_j,\label{pbm}\\
\left[\alpha,\beta\right]_{\rm m}&=\alpha_i\mc{J}^{ik}\mc{J}^{jk}\beta_j.\label{dbm}
\end{align}
\end{subequations}
Here, $\pb{}{}_{\rm m}$ is the microscopic Poisson bracket encountered in the previous section,
while $\left[,\right]_{\rm m}$ is the microscopic dissipative bracket.
Notice that the tensor $\mc{J}^{ik}\mc{J}^{jk}$ is symmetric and non-negative
as required by the definition. 
Hence, the stochastic equation \eqref{Xi2} can be cast in bracket notation as below:
\begin{equation}
\dot{X}^i=\pb{x^i}{H+\delta H}_{\rm m}-\left[x^i,\gamma H\right]_{\rm m}.\label{Xdotbm}
\end{equation}
In the next section, it will be shown that the microscopic metriplectic bracket \eqref{mmb}
induce a macroscopic metriplectic bracket for the time evolution of the distribution function $f$
that is consistent with conservation of energy and the second law of thermodynamics. 

\section{Dissipative Brackets for the Fokker-Planck Equation}

The stochastic equation of motion \eqref{Xi2} 
can be translated into a Fokker-Planck equation for
the distribution function $f$.
Details can be found in \cite{Sato18,Sato19}. 
The result is
\begin{equation}
\frac{\p f}{\p t}=\frac{\p}{\p x^i}\left[-\mc{J}^{ij}H_jf+\gamma\mc{J}^{ik}\mc{J}^{jk}H_jf+\frac{1}{2}D\mc{J}^{ik}\frac{\p}{\p x^j}\lr{\mc{J}^{jk}f}\right].\label{FPE1}
\end{equation}
Here, the Stratonovich convention for stochastic integration was adopted. 
Since, by construction, $f$ is the distribution function on the invariant measure, i.e. $J=1$, from equation \eqref{IM} we have 
\begin{equation}
\mc{J}^{ij}_i=0,~~~~j=1,...,n.\label{IM2}
\end{equation} 
Hence, equation \eqref{FPE1} can be further simplified to
\begin{equation}
\frac{\p f}{\p t}=-\mc{J}^{ij}H_jf_i+\frac{1}{2}D\mc{J}^{ik}\frac{\p}{\p x^i}\left[f\mc{J}^{jk}\frac{\p}{\p x^j}\lr{\log f+\beta H}\right],\label{FPE2}
\end{equation}
where we defined the spatial constant (inverse temperature)
\begin{equation}
\beta=\frac{2\gamma}{D}.
\end{equation}
Notice that $\beta=\beta\left[f\right]$ (and therefore $\gamma$) 
is a functional of $f$.  
Thus, when the system is outside of equilibrium, $\beta$ can be a function of time. These aspects will be discussed later in relation to conservation of energy.  
From equation \eqref{FPE2} one recognizes 
a candidate Poisson bracket,
\begin{equation}
\left\{F,G\right\}=\int_{\Omega}f\frac{\p}{\p x^i}\lr{\frac{\delta F}{\delta f}}_{\beta}\mc{J}^{ij}\frac{\p}{\p x^j}\lr{\frac{\delta G}{\delta f}}_{\beta}dV,~~~~F,G\in \mc{X}^\ast.\label{PBM}
\end{equation}
and a candidate dissipative bracket
\begin{equation}
\left[F,G\right]=\frac{D}{2}\int_{\Omega}f\frac{\p}{\p x^i}\lr{\frac{\delta F}{\delta f}}_{\beta}\mc{J}^{ik}\mc{J}^{jk}\frac{\p}{\p x^j}\lr{\frac{\delta G}{\delta f}}_{\beta}dV,~~~~F,G\in \mc{X}^\ast.\label{DBM}
\end{equation}
In this notation, $\mc{X}^\ast$ represents the set of differentiable functionals $F:\mc{X}\rightarrow\mathbb{R}$, 
with $\mc{X}$ a vector space over $\mathbb{R}$, 
while the lower index $\beta$ stands for variations
at fixed temperature. The restriction of variations $\delta f$ 
to those that leave $\beta$ unchanged is needed to enforce conservation 
of energy. In the following, we shall omit the lower index $\beta$
to simplify the notation. 
In order to validate \eqref{PBM} and \eqref{DBM}
we must verify that the defining properties of 
Poisson and dissipative brackets are satisfied, and 
that these brackets generate the Fokker-Planck equation \eqref{FPE2}
by suitable choice of total energy, entropy, and boundary conditions. 

First, consider the candidate bracket \eqref{PBM}.
This bracket defines a Poisson algebra in $\mc{X}^\ast$ over 
the field of real numbers $\mathbb{R}$ 
provided that it satisfies the following axioms:
\begin{subequations}
\begin{align}
&\pb{aF+bG}{H}=a\pb{F}{H}+b\pb{G}{H},~~~~\pb{H}{aF+bG}=a\pb{H}{F}+b\pb{H}{G},\label{P1}\\
&\pb{F}{F}=0,\label{P2}\\
&\pb{F}{G}=-\pb{G}{F},\label{P3}\\
&\pb{FG}{H}=F\pb{G}{H}+\pb{F}{H}G,\label{P4}\\
&\pb{F}{\pb{G}{H}}+\car=0,\label{P5}
\end{align}
\end{subequations}
forall $a,b\in\mathbb{R}$ and $F,G,H\in\mc{X}^\ast$.
These axioms are bilinearity, alternativity, antisymmetry 
(which follows from the first two axioms), 
Leibniz rule, and Jacobi identity respectively. 
\eqref{P1} is trivially satisfied.
Alternativity \eqref{P2} and antisymmetry \eqref{P3} follow
from the antisymmetry of the Poisson tensor $\mc{J}$.
The Leibniz rule \eqref{P4} is also satisfied. Indeed, 
\begin{equation}
\begin{split}
\pb{FG}{H}&=\int_{\Omega}f\frac{\p}{\p x^i}\lr{F\frac{\delta G}{\delta f}+G\frac{\delta F}{\delta f}}\mc{J}^{ij}\frac{\p}{\p x^j}\lr{\frac{\delta H}{\delta f}}dV\\&=
F\int_{\Omega}f\frac{\p}{\p x^i}\lr{\frac{\delta G}{\delta f}}\mc{J}^{ij}\frac{\p}{\p x^j}\lr{\frac{\delta H}{\delta f}}dV+
G\int_{\Omega}f\frac{\p}{\p x^i}\lr{\frac{\delta F}{\delta f}}\mc{J}^{ij}\frac{\p}{\p x^j}\lr{\frac{\delta H}{\delta f}}dV=F\pb{G}{H}+\pb{F}{H}G.
\end{split}
\end{equation}
Next, consider the Jacobi identity \eqref{P5}:
\begin{equation}
\begin{split}
\pb{F}{\pb{G}{H}}+\car&=\int_{\Omega}f\frac{\p}{\p x^i}\lr{\frac{\delta F}{\delta f}}\mc{J}^{ij}\frac{\p}{\p x^j}\lr{\frac{\delta\pb{G}{H}}{\delta f}}dV+\car
\\&=\int_{\Omega}f\frac{\p}{\p x^i}\lr{\frac{\delta F}{\delta f}}\mc{J}^{ij}\frac{\p}{\p x^j}\left[\frac{\delta}{\delta f}\int_{\Omega}f\frac{\p}{\p x^m}\lr{\frac{\delta G}{\delta f}}\mc{J}^{mp}\frac{\p}{\p x^p}\lr{\frac{\delta H}{\delta f}}dV\right]dV+\car.\label{JI1}
\end{split}
\end{equation}
It is convenient to simplify the notation by specifying derivatives with lower indexes, e.g. $F_{if}=\frac{\p}{\p x^i}\lr{\frac{\delta F}{\delta f}}$.
Equation \eqref{JI1} becomes
\begin{equation}
\begin{split}
\pb{F}{\pb{G}{H}}+\car&=
\int_{\Omega}fF_{if}\mc{J}^{ij}\frac{\p}{\p x^j}\lr{\frac{\delta}{\delta f}\int_{\Omega}f\mc{J}^{mp}G_{mf}H_{pf}dV}dV
+\car.\label{JI2}
\end{split}
\end{equation}
Next, observe that
\begin{equation}
\delta\int_{\Omega}f\mc{J}^{mp}G_{mf}H_{pf}dV=\int_{\Omega}{\delta f \mc{J}^{mp}G_{mf}H_{pf}}dV+\int_{\Omega}f\mc{J}^{mp}\delta\lr{G_{mf}H_{pf}}dV.
\end{equation}
The second term on the right-hand side does not contribute
to the Jacobi identity because it contains second order functional derivatives.
The cancellation of terms involving second order derivatives can be
easily verified for the Jacobi identity in finite dimensions.
Equation \eqref{JI2} thus reduces to
\begin{equation}
\pb{F}{\pb{G}{H}}+\car=\int_{\Omega}f\left[\mc{J}^{ij}\mc{J}^{mp}_jF_{if}G_{mf}H_{pf}+\mc{J}^{ij}\mc{J}^{mp}F_{if}\lr{G_{jmf}H_{pf}+G_{mf}H_{jpf}}\right]dV+\car.
\end{equation}
The first term in the integrand vanishes due to the Jacobi identity satisfied 
by the Poisson operator $\mc{J}^{ij}$,
\begin{equation}
\mc{J}^{ij}\mc{J}_j^{mp}+\mc{J}^{mj}\mc{J}_j^{pi}+\mc{J}^{pj}\mc{J}_j^{im}=0~~~~i,m,p=1,...,n.
\end{equation}
The second term can be rearranged as follows:
\begin{equation}
\begin{split}
\mc{J}^{ij}\mc{J}^{mp}F_{if}\lr{G_{jmf}H_{pf}+G_{mf}H_{jpf}}+\car=
&\frac{1}{2}\mc{J}^{ij}\mc{J}^{mp}F_{if}G_{jmf}H_{pf}
+\frac{1}{2}\mc{J}^{im}\mc{J}^{jp}F_{if}G_{jmf}H_{pf}\\
&+\frac{1}{2}\mc{J}^{ij}\mc{J}^{mp}F_{if}G_{mf}H_{jpf}
+\frac{1}{2}\mc{J}^{ip}\mc{J}^{mj}F_{if}G_{mf}H_{jpf}\\
&+\frac{1}{2}\mc{J}^{ij}\mc{J}^{mp}G_{if}H_{jmf}F_{pf}
+\frac{1}{2}\mc{J}^{im}\mc{J}^{jp}G_{if}H_{jmf}F_{pf}\\
&+\frac{1}{2}\mc{J}^{ij}\mc{J}^{mp}G_{if}H_{mf}F_{jpf}
+\frac{1}{2}\mc{J}^{ip}\mc{J}^{mj}G_{if}H_{mf}F_{jpf}\\
&+\frac{1}{2}\mc{J}^{ij}\mc{J}^{mp}H_{if}F_{jmf}G_{pf}
+\frac{1}{2}\mc{J}^{im}\mc{J}^{jp}H_{if}F_{jmf}G_{pf}\\
&+\frac{1}{2}\mc{J}^{ij}\mc{J}^{mp}H_{if}F_{mf}G_{jpf}
+\frac{1}{2}\mc{J}^{ip}\mc{J}^{mj}H_{if}F_{mf}G_{jpf}.
\end{split}
\end{equation}
The terms on the right-hand side cancel in pairs.
For example, 
\begin{equation}
\begin{split}
&\frac{1}{2}\mc{J}^{ij}\mc{J}^{mp}F_{if}G_{jmf}H_{pf}+
\frac{1}{2}\mc{J}^{ip}\mc{J}^{mj}H_{if}F_{mf}G_{jpf}=\\
&\frac{1}{2}\mc{J}^{ij}\mc{J}^{mp}F_{if}G_{jmf}H_{pf}+
\frac{1}{2}\mc{J}^{im}\mc{J}^{pj}H_{if}F_{pf}G_{jmf}=\\
&\frac{1}{2}\mc{J}^{ij}\mc{J}^{mp}F_{if}G_{jmf}H_{pf}+
\frac{1}{2}\mc{J}^{pm}\mc{J}^{ij}H_{pf}F_{if}G_{jmf}=0.
\end{split}
\end{equation}
Hence, the Jacobi identity is satisfied, and \eqref{PBM} is
a Poisson bracket. 
We refer the reader to \cite{Marsden} for additional considerations  
on the natural Poisson structure on the dual space of a Poisson algebra.  

Similarly, the candidate bracket \eqref{DBM} defines a dissipative algebra
in $\mc{X}^\ast$ over the field of real numbers $\mathbb{R}$ 
provided that it satisfies the following axioms:
\begin{subequations}
\begin{align}
&\db{aF+bG}{H}=a\db{F}{H}+b\db{G}{H},~~~~\db{H}{aF+bG}=a\db{H}{F}+b\db{H}{G},\label{D1}\\
&\db{F}{F}\geq 0,\label{D2}\\
&\db{F}{G}=\db{G}{F},\label{D3}\\
&\db{FG}{H}=F\db{G}{H}+\db{F}{H}G,\label{D4}
\end{align}
\end{subequations}
forall $a,b\in\mathbb{R}$ and $F,G,H\in\mc{X}^\ast$.
These axioms are bilinearity, non-negativity, symmetry, 
and Leibniz rule respectively. 
Bilinearity \eqref{D1} is trivially satisfied.
Non-negativity \eqref{D2} can be verified as below
\begin{equation}
\db{F}{F}=\frac{D}{2}\int_{\Omega}f\frac{\p}{\p x^i}\lr{\frac{\delta F}{\delta f}}\mc{J}^{ik}\mc{J}^{jk}\frac{\p}{\p x^j}\lr{\frac{\delta F}{\delta f}}dV=
\frac{D}{2}\sum_{i=1}^n\int_{\Omega}f\left[\mc{J}^{ij}\frac{\p}{\p x^j}\lr{\frac{\delta F}{\delta f}}\right]^2dV\geq 0.
\end{equation}
Here, we used the fact that $f\geq 0$ since it is a distribution function.
Symmetry \eqref{D3} 
follows from the symmetry of the tensor $\mc{J}^{ik}\mc{J}^{jk}$. 
The Leibniz rule \eqref{D4} can be checked with the same calculation 
used for the Poisson bracket.
An additional axiom may be considered for the dissipative bracket
that replaces the Jacobi identity of the Poisson case.
For example, one may require the dissipative bracket
to originate from a Poisson bracket (as in the case under consideration). 
We will refer to a dissipative bracket induced from a Poisson bracket as a natural
dissipative bracket. 
Some further geometrical aspects pertaining to dissipative
algebras will be discussed in section 5. 

In order to express the Fokker-Planck equation \eqref{FPE2} in terms of the new metriplectic brackets,
we consider the following macroscopic observables:
\begin{subequations}\label{obs}
\begin{align}
&N=\int_{\Omega}f\,dV=1,\\
&E=\int_{\Omega}fH\,dV,\\
&\mf{C}^k=\int_{\Omega}fC^k\,dV,~~~~k=1,...,m,\\
&S=-\int_{\Omega}{f\log{f}\,dV},\\
&\Sigma=S-\alpha N-\beta E-\mu_k\mf{C}^k.
\end{align}
\end{subequations}
As in the previous section, $N$ is the total particle number, $E$ the total energy,
$\mf{C}^k$ the total $k$th Casimir invariant, and $S$ Shannon's information entropy measure
for the distribution function on the invariant measure. We shall refer to the quantity $\Sigma$ 
as the constrained entropy of the system. 
Next, we express the Fokker-Planck equation \eqref{FPE2} in divergence form
\begin{equation}
\frac{\p f}{\p t}=-\frac{\p}{\p x^i}\lr{fZ^i},
\end{equation}
with
\begin{equation}
Z^i=\mc{J}^{ij}H_j-\frac{1}{2}D\mc{J}^{ik}\mc{J}^{jk}\frac{\p}{\p x^j}\lr{\log{f}+\beta H},~~~~i=1,...,n,
\end{equation}
the $i$th component of Fokker-Planck velocity $\bol{Z}=Z^i\p_i$. 
To ensure the thermodynamic closure of the system we impose the
boundary conditions
\begin{equation}
\dot{\bol{x}}\cdot\bol{n}=0,~~~~\bol{Z}\cdot\bol{n}=0~~~~{\rm on}~~\p\Omega,\label{BC}
\end{equation}
with $\bol{n}$ the unit outward normal to $\p\Omega$. 
The first boundary condition applies to the dynamics of the unperturbed
particle $\dot{\bol{x}}=\mc{J}^{ij}H_j\p_i$, and can be satisfied, for example, 
by choosing $\p\Omega$ to be a level set of the Hamiltonian $H$. Then, 
$\bol{n}=\nabla H/\abs{\nabla H}$ and $\dot{\bol{x}}\cdot\bol{n}=\mc{J}^{ij}H_iH_j/\abs{\nabla H}=0$ on $\p\Omega$.  
The second boundary condition applies to the dynamics of the ensemble. 
When $\bol{n}=\nabla H/\abs{\nabla H}$, we have
\begin{equation}
\bol{Z}\cdot\bol{n}=-\frac{D}{2\abs{\nabla H}}\mc{J}^{ik}\mc{J}^{jk}H_i\frac{\p}{\p x^j}\lr{\log f+\beta H}=0~~~~{\rm on}~~\p\Omega.
\end{equation}
This equation gives the boundary condition for the distribution function $f$. 
Under the hypothesis above and appropriate choice of the spatial constant $\beta=2\gamma/D$, we claim that
\begin{subequations}
\begin{align}
&\frac{\p f}{\p t}=\lr{f,E,\Sigma},\label{fmb}\\
&\frac{dE}{dt}=\lr{E,E,\Sigma}=0,\label{Emb}\\
&\frac{d\Sigma}{dt}=\lr{\Sigma,E,\Sigma}\geq 0,\label{Smb}
\end{align}
\end{subequations}
where $\lr{\circ,E,\Sigma}=\mb{\circ}{E}{\Sigma}$ denotes the metriplectic bracket.
To see this, first we need to explain how boundary conditions \eqref{BC} are applied  because the
coordinate system $\bol{x}=\lr{x^1,...,x^n}$ is curvilinear and the application of the divergence theorem is not straightforward.
Let $\bol{x}_c=\lr{x_c^1,...,x_c^n}$ denote a Cartesian coordinate system in $\Omega$ with tangent
basis $\bol{\p}_c=\lr{\p_{c1},...,\p_{cn}}$ and such that $dV=J_c dV_c$,
with $dV_c=dx^1_c\w ... \w dx^n_c$. Consider a vector field $\bol{v}=v^i\p_i=v_c^i\p_{ci}$. We have 
\begin{equation}
\int_{\Omega}\frac{\p v^i}{\p x^i}\,dV=\int_{\Omega}\mf{L}_{\bol{v}}dV=\int_{\Omega}\mf{L}_{\bol{v}}J_cdV_c=\int_{\Omega}\frac{\p}{\p x_c^i}\lr{J_c v_c^i}\,dV_c=\int_{\p\Omega}\bol{v}\cdot\bol{n}\,J_cdS_c=\int_{\p\Omega}\bol{v}\cdot\bol{n}\, dS.
\end{equation}
In the calculation above, we used the divergence theorem in the penultimate passage, 
and introduced the notation $dS=J_cdS_c$, with $dS_c$ 
the surface element on $\p\Omega$. Hence, if $\bol{v}\cdot\bol{n}=0$ on $\p\Omega$, the integral vanishes.   

Now consider conservation of energy \eqref{Emb}. Using \eqref{FPE1} and applying boundary conditions,   
\begin{equation}
\frac{dE}{dt}=\int_{\Omega}f_tHdV=\int_{\Omega}fZ^iH_idV=-\frac{D}{2}\int_{\Omega}fH_i\mc{J}^{ik}\mc{J}^{jk}\frac{\p}{\p x^j}\lr{\log{f}+\beta H}dV.\label{dEdt}
\end{equation}
For the right-hand side to vanish, the following must hold:
\begin{equation}
\beta=-\frac{\int_{\Omega}H_i\mc{J}^{ik}\mc{J}^{jk}f_jdV}{\int_{\Omega}fH_i\mc{J}^{ik}\mc{J}^{jk}H_jdV}=-\frac{\sum_{i=1}^n\langle\dot{x}^i\mc{J}^{ij}\p_j\log{f}\rangle}{\sum_{i=1}^n\langle\lr{\dot{x}^i}^2\rangle}.\label{beta}
\end{equation}
In this notation, the angle bracket denotes ensemble averaging.
This relationship defines the inverse temperature $\beta$ at each time $t$.
If the system is sufficiently close to equilibrium, 
$\beta$ can be replaced by its equilibrium value and treated as a space-time constant. 
However, in general $\beta=\beta\left[f\right]$ is a functional of $f$.
Returning to \eqref{Emb}, observe that $\pb{E}{E}=0$. Hence, 
\begin{equation}
\begin{split}
\lr{E,E,\Sigma}=\db{E}{\Sigma}=&\frac{D}{2}\int_{\Omega}fH_i\mc{J}^{ik}\mc{J}^{jk}\frac{\p}{\p x^j}\lr{-\log{f}-1-\alpha-\beta H-\mu_p C^p}\\
=&-\frac{D}{2}\int_{\Omega}fH_i\mc{J}^{ik}\mc{J}^{jk}\frac{\p}{\p x^j}\lr{\log{f}+\beta H}dV,\label{mdedt}
\end{split}
\end{equation}
which is exactly \eqref{dEdt}. This gives conservation of energy \eqref{Emb} when
the inverse temperature is given by \eqref{beta}.
Notice that, when evaluating the macroscopic brackets, 
variations are carried out at fixed $\beta$. 
Furthermore, in the last passage of \eqref{mdedt}, we used the fact that $N$ is a Casimir invariant of the macroscopic brackets,
and that the macroscopic brackets inherit the Casimir invariants $\mf{C}^k$, $k=1,...,m$, from microscopic dynamics, i.e. 
\begin{equation}
\pb{\mf{C}^k}{E}=\db{\mf{C}^k}{\Sigma}=\lr{\mf{C}^{k},E,\Sigma}=0~~~~\forall E,\Sigma,~~k=1,...,m.
\end{equation}

Next, consider equation \eqref{fmb}. We have
\begin{equation}
\begin{split}
\lr{f,E,\Sigma}=&\int_{\Omega}f\frac{\p}{\p x^i}\left[\delta\lr{\bol{x}-\bol{x}'}\right]\mc{J}^{ij}\frac{\p H}{\p x^j}\,dV\\&+
\frac{D}{2}\int_{\Omega}f\frac{\p}{\p x^i}\left[\delta\lr{\bol{x}-\bol{x}'}\right]\mc{J}^{ik}\mc{J}^{jk}\frac{\p}{\p x^j}\lr{-\log{f}-1-\alpha-\beta H-\mu_p C^p}dV\\
=&\int_{\Omega}f\frac{\p}{\p x^i}\left[\delta\lr{\bol{x}-\bol{x}'}\right]\left[\mc{J}^{ij}H_j-\frac{D}{2}\mc{J}^{ik}\mc{J}^{jk}\frac{\p}{\p x^j}\lr{\log{f}+\beta H}\right]dV.
\end{split}
\end{equation}
Using the boundary condition $\bol{Z}\cdot\bol{n}=0$ on $\p\Omega$ to eliminate surface integrals, 
and recalling that $\mc{J}^{ij}_i=0$, $j=1,...,n$, 
integration by parts gives
\begin{equation}
\lr{f,E,\Sigma}=\frac{\p}{\p x^i}\left[-\mc{J}^{ij}H_jf+\gamma\mc{J}^{ik}\mc{J}^{jk}H_jf+\frac{1}{2}D\mc{J}^{ik}\frac{\p}{\p x^j}\lr{\mc{J}^{jk}f}\right],
\end{equation}
which is the right-hand side of the Fokker-Planck equation \eqref{FPE1} as desired.

Consider now the entropy law \eqref{Smb}. 
From the Fokker-Planck equation \eqref{FPE1}, 
boundary conditions \eqref{BC}, conservation of
total particle number $N$, conservation of total energy $E$, 
and conservation of total Casimir invariants $\mf{C}^k$, $k=1,...,m$,
it follows that
\begin{equation}
\begin{split}
\frac{d\Sigma}{dt}=\frac{dS}{dt}=&
-\int_{\Omega}{f_t\lr{\log{f}+1}}\,dV=
\int_{\Omega}\frac{\p}{\p x^i}\left[fZ^i\lr{\log{f}+1}\right]dV
-\int_{\Omega}Z^if_i\,dV\\
=&\int_{\p\Omega}f\lr{\log{f}+1}\bol{Z}\cdot\bol{n}\,dS
-\int_{\Omega}\left[\mc{J}^{ij}H_jf_i-\frac{1}{2}D\mc{J}^{ik}\mc{J}^{jk}f_i\frac{\p}{\p x^j}\lr{\log{f}+\beta H}\right]dV\\
=&-\int_{\Omega}\frac{\p}{\p x^i}\lr{\mc{J}^{ij}H_jf}dV
+\frac{1}{2}D\int_{\Omega}f\mc{J}^{ik}\mc{J}^{jk}\frac{\p}{\p x^i}\lr{\log{f}+\beta H}\frac{\p}{\p x^j}\lr{\log{f}+\beta H}dV\\&-\frac{1}{2}D\int_{\Omega}f\mc{J}^{ik}\mc{J}^{jk}\beta H_i\frac{\p}{\p x^j}\lr{\log{f}+\beta H}dV\\
=&-\int_{\p\Omega}f\dot{\bol{x}}\cdot\bol{n}\,dS+\frac{1}{2}D\sum_{i=1}^n\int_{\Omega}f\left[\mc{J}^{ij}\frac{\p}{\p x^j}\lr{\log{f}+\beta H}\right]^2dV
-\beta \frac{dE}{dt}\\
=&\frac{1}{2}D\sum_{i=1}^n\int_{\Omega}f\left[\mc{J}^{ij}\frac{\p}{\p x^j}\lr{\log{f}+\beta H}\right]^2dV\geq 0.\label{dSdt}
\end{split}
\end{equation}
In this calculation we used 
the hypothesis that $D\geq 0$ and $f\geq 0$.
Notice that, if $f>0$ at thermodynamic equilibrium, equation \eqref{dSdt}
implies that
\begin{equation}
\lim_{t\rightarrow \infty}\mc{J}^{ij}\frac{\p}{\p x^j}\lr{\log{f}+\beta H}=0.
\end{equation}
This expression gives the expected equilibrium state \eqref{feq}.
In a similar way, one sees that
\begin{equation}
\begin{split}
\lr{\Sigma,E,\Sigma}=&\int_{\Omega}f\frac{\p}{\p x^i}\lr{-\log{f}-\beta H}\mc{J}^{ij}
\frac{\p H}{\p x^j}\,dV+\frac{D}{2}\int_{\Omega}f\frac{\p}{\p x^i}\lr{-\log{f}-\beta H}\mc{J}^{ik}\mc{J}^{jk}\frac{\p}{\p x^j}\lr{-\log{f}-\beta H}dV\\
=&-\int_{\p\Omega}f\dot{\bol{x}}\cdot\bol{n}\,dS
+\frac{1}{2}D\sum_{i=1}^n\int_{\Omega}f\left[\mc{J}^{ij}\frac{\p}{\p x^j}\lr{\log{f}+\beta H}\right]^2dV\\=&
\frac{1}{2}D\sum_{i=1}^n\int_{\Omega}f\left[\mc{J}^{ij}\frac{\p}{\p x^j}\lr{\log{f}+\beta H}\right]^2dV,
\end{split}
\end{equation}
which is exactly the rate of change obtained in \eqref{dSdt}.

Finally, observe that, since $\pb{\Sigma}{E}=0$ forall $E$,
it is possible to use $\Sigma$ as a single generating function
by redefining the Poisson bracket as $\pb{}{}\rightarrow -\beta^{-1}\pb{}{}$. Then,
\begin{equation}
\frac{\p f}{\p t}=\lr{f,\Sigma}.
\end{equation}
However, notice that $E$ is not a Casimir invariant of the dissipative part of the bracket.

\section{Fokker-Planck Equation for Infinite Dimensional Hamiltonian Systems}

Aim of the present section is to 
provide some examples of how the formalism
discussed above can be applied
to construct a Fokker-Planck equation with a metriplectic structure
and a thermodynamic equilibrium of the type \eqref{feq}
for infinite dimensional Hamiltonian systems. 
The main hurdle in generalizing the theory
from finite dimensions to infinite dimensions
is represented by the notion of volume  
in arbitrary function spaces. 
Given a function space $\mc{X}$ 
with elements $\rho:\Omega\rightarrow\mathbb{R}$, 
suppose that there exists an orthonormal expansion
\begin{equation}
\rho=\sum_{i=-\infty}^{\infty}\lr{\rho,u_i}u_i=\psi^i u_i~~~~\forall \rho\in\mc{X},\label{expa}
\end{equation}
with $\lr{u_0,u_{1},u_2,...,u_{-1},u_{-2},...}$ an orthonormal basis of $\mc{X}$
and 
\begin{equation}
\lr{\rho,\sigma}=\int_{\Omega}\rho\sigma\,dV,~~~~\rho,\sigma\in\mc{X},
\end{equation}
the inner product on $\mc{X}$. 
Further assume that the evolution of $\rho$ is governed by an infinite dimensional
Hamiltonian system 
\begin{equation}
\frac{\p \rho}{\p t}=\mc{J}\frac{\delta H}{\delta \rho},\label{ftinf}
\end{equation}
with $\mc{J}:T^\ast\mc{X}\rightarrow T\mc{X}$ the Poisson operator and $H\in\mc{X}^\ast$ the Hamiltonian function. Now equation \eqref{ftinf} 
can be rewritten as an equivalent Hamiltonain system for  
the coefficients $\psi^i=\lr{\rho,u_i}$ of the expansion \eqref{expa}:
\begin{equation}
\frac{\p \psi}{\p t}=\hat{\mc{J}}\frac{\p H}{\p\psi}.\label{psit}
\end{equation}
Here, $\psi=\lr{\psi^0,\psi^1,\psi^2,...,\psi^{-1},\psi^{-2},...}$
is an infinite dimensional vector, $\hat{\mc{J}}$ the
associated Poisson operator (an infinite dimensional matrix depending on $\psi$), and $H=H\lr{\psi}$ the Hamiltonian as a function of the new variables $\psi$.
Notice that functional derivatives have been replaced by partial derivatives in \eqref{psit}. 
Then, the invariant measure of the system has the form
\begin{equation}
J dV= J d\psi^0\w d\psi^1\w d\psi^2 \w...\w d\psi^{-1}\w d\psi^{-2} \w ...~,
\end{equation}
where $J$ is given as the solution of
\begin{equation}
\p_i\lr{J\hat{\mc{J}}^{ij}}=0,~~~~j=0,1,2,...,-1,-2,...~.
\end{equation}
Performing a change of variables $\psi\rightarrow\theta$ such that
\begin{equation}
dV'=d\theta^0\w d\theta^1\w d\theta^2\w ...\w d\theta^{-1}\w d\theta^{-2}\w ...=JdV,
\end{equation}
define the distribution function $f=f\lr{\theta}$ of the new variables $\theta\lr{\psi}$
with respect to the invariant measure $dV'$. 
Then, the Fokker-Planck equation associated with the original system \eqref{ftinf}
can be written within the metriplectic formalism as
\begin{equation}
\frac{\p f}{\p t}=\lr{f,E,\Sigma}.
\end{equation}
Here, the metriplectic bracket $\lr{\circ,E,\Sigma}=\pb{\circ}{E}+\db{\circ}{\Sigma}$, the total energy $E$, and the constrained entropy $\Sigma$
are defined as in \eqref{PBM}, \eqref{DBM}, and \eqref{obs}. Notice that, however, all the quantities 
are written in terms of $f\lr{\theta}$, $H\lr{\theta}$, $\hat{\mc{J}}'\lr{\theta}$ (the Poisson operator in the new variables),
and $\hat{C}^k\lr{\theta}$, $k=1,2,...$, (the Casimir invariants of $\hat{\mc{J}}'$).
Below we provide an example of the procedure described above in the context of plasma physics.
We also remark that this approach is not restricted to orthonormal expansions,
but applies to discrete basis in general.

\subsection{Fokker-Planck equation for the Charney-Hasegawa-Mima equation}

The Charney-Hasegawa-Mima equation 
\begin{equation}
\lr{1-\Delta}\phi_t=\phi_x\lr{\Delta\phi_y+\lambda_y}-\phi_y\lr{\Delta\phi_x+\lambda_x},\label{CHM}
\end{equation}
is a nonlinear partial differential equation
describing the time evolution of a function $\phi=\phi\lr{x,y,t}$ in 2 dimensional space. 
The function $\lambda=\lambda\lr{x,y}$ is given and characterizes certain physical properties of the system.
Equation \eqref{CHM} occurs in the description of plasma turbulence, where the function $\phi$ 
represents the electric potential, and in the context of geophysical fluid dynamics,
where the function $\phi$ represents the stream function. 
The noncanonical Hamiltonian structure of \eqref{CHM} is known and can be found, for example, in \cite{Weinstein,Tassi}. 
Usually, the right-hand side of equation \eqref{CHM} is expressed in bracket notation
as $\left[\phi,\Delta\phi+\lambda\right]$, with $\left[\alpha,\beta\right]=\alpha_x\beta_y-\alpha_y\beta_x$
for any pair of differentiable functions $\alpha$ and $\beta$.
Here, we do not use such notation to avoid confusion with the dissipative bracket.

To simplify the calculations, we assume that $\lambda=cy$, with $c\in\mathbb{R}$.
This requirement, which implies that $\lambda$ varies uniformly in one direction, is often found in applications (see e.g. \cite{Swaters}).
Then, equation \eqref{CHM} reduces to
\begin{equation}
\lr{1-\Delta}\phi_t=\phi_x\lr{\Delta\phi_y+c}-\phi_y\Delta\phi_x.\label{CHM2}
\end{equation}
We look for periodic solutions of \eqref{CHM2} in both the $x$ and $y$ directions with period $2\pi$. 
The spatial domain of the function $\phi$ is taken to be $\Omega=\left[-\pi,\pi\right]^2$. 
Then, the decomposition \eqref{expa} can be obtained in terms of a Fourier series
\begin{equation}
\phi=\sum_{n,m=-\infty}^{+\infty}\phi^{nm}e^{i\lr{nx+my}}.\label{FS}
\end{equation} 
Substituting the Fourier series \eqref{FS} into \eqref{CHM2}, 
one obtains a system of equations for the Fourier coefficients $\phi_{nm}$:
\begin{equation}
\dot{\phi}^{nm}=\frac{1}{1+n^2+m^2}\sum_{p,q=-\infty}^{+\infty}\left\{icn\delta_{np}\delta_{mq}+\lr{mp-nq}\left[\lr{n-p}^2+\lr{m-q}^2\right]\phi^{n-p\,m-q}\right\}\phi^{pq}.\label{CHMF}
\end{equation}
Notice that
\begin{equation}
\sum_{p,q=-\infty}^{+\infty}\lr{mp-nq}\phi^{n-p\,m-q}\phi^{pq}=0.
\end{equation}
This is because each term in the summation corresponding to a given pair $\lr{p,q}$ 
cancels with the pair $\lr{p',q'}=\lr{n-p,m-q}$. Therefore, equation \eqref{CHMF} can be rewritten as
\begin{equation}
\dot{\phi}^{nm}=\frac{1}{1+n^2+m^2}\sum_{p,q=-\infty}^{+\infty}\left\{icn\delta_{np}\delta_{mq}+\lr{mp-nq}\left[1+\lr{n-p}^2+\lr{m-q}^2\right]\phi^{n-p\,m-q}\right\}\phi^{pq}.\label{CHMF}
\end{equation}
This form will be useful when proving the Jacobi identity.
The energy of the system is given by
\begin{equation}
H=\frac{1}{2}\int_{\Omega}\lr{\phi^2+\abs{\nabla\phi}^2}dxdy=2\pi^2\sum_{n,m=-\infty}^{+\infty}\lr{1+n^2+m^2}\abs{\phi^{nm}}^2.\label{HCHM}
\end{equation}
In the last passage we used the fact that, since $\phi$ is real, $\phi^{nm\ast}=\phi^{-n\,-m}$, with $\ast$ denoting complex conjugation.
It follows that
\begin{equation}
\frac{\p H}{\p \phi^{uv}}=
\begin{cases}
4\pi^2\lr{1+u^2+v^2}\phi^{uv\ast}~~~~&{\rm if}~~\lr{u,v}\neq\lr{0,0},\\
2\pi^2\phi^{00}~~~~&{\rm if}~~\lr{u,v}=\lr{0,0}.
\end{cases}
.\label{dHdphiuv}
\end{equation}
Substituting \eqref{dHdphiuv} into \eqref{CHMF} gives
\begin{equation}
\begin{split}
\dot{\phi}^{nm}=&\sum_{p,q=-\infty}^{+\infty}\frac{
icn\delta_{np}\delta_{mq}+\lr{mp-nq}\left[1+\lr{n-p}^2+\lr{m-q}^2\right]\phi^{n-p\,m-q}
}{4\pi^2\lr{1+n^2+m^2}\lr{1+p^2+q^2}}\frac{\p H}{\p \phi^{pq\ast}}\\
=&\sum_{p,q=-\infty}^{+\infty}\frac{
icn\delta_{n\,-p}\delta_{m\,-q}+\lr{nq-mp}\left[1+\lr{n+p}^2+\lr{m+q}^2\right]\phi^{n+p\,m+q}
}{4\pi^2\lr{1+n^2+m^2}\lr{1+p^2+q^2}}\frac{\p H}{\p\phi^{pq}}.\label{CHMF2}
\end{split}
\end{equation}
Next, define the vector with components $\phi^i=\phi^{i_1i_2}$ given by   
$\lr{\phi^0,\phi^1,\phi^2...}=\lr{\phi^{00},\phi^{01},\phi^{02},...}$.
Then, \eqref{CHMF2} can be cast in the form
\begin{equation}
\dot{\phi}^i=\mc{J}^{ij}H_j,
\end{equation}
where $H_j=\p H/\p\phi^j$, and we identified the candidate Poisson operator
\begin{equation}
\mc{J}^{ij}=\mc{B}^{ij}+\mc{C}^{ij}\phi^{i+j}.\label{JCHM}
\end{equation}
Here, $\mc{B}^{ij}$ and $\mc{C}^{ij}$ are constants depending on $i_1$, $i_2$, $j_1$, and $j_2$,
\begin{equation}
\mc{B}^{ij}=\frac{ici_1\delta_{i_1\,-j_1}\delta_{i_2\,-j_2}}{4\pi^2\lr{1+i_1^2+i_2^2}\lr{1+j_1^2+j_2^2}},~~~~\mc{C}^{ij}=\frac{
\lr{i_1j_2-i_2j_1}\left[1+\lr{i_1+j_1}^2+\lr{i_2+j_2}^2\right]}{4\pi^2\lr{1+i_1^2+i_2^2}\lr{1+j_1^2+j_2^2}},
\end{equation}
and $\phi^{i+j}=\phi^{n+p\,m+q}$.
The operator $\mc{J}^{ij}$ is antisymmetric. Indeed,
\begin{equation}
\begin{split}
\mc{J}^{ji}=&\frac{icj_1\delta_{j_1\,-i_1}\delta_{j_2\,-i_2}+\lr{j_1i_2-j_2i_1}\left[1+\lr{j_1+i_1}^2+\lr{j_2+i_2}^2\right]\phi^{j_1+i_1\,j_2+i_2}}{4\pi^2\lr{1+j_1^2+j_2^2}\lr{1+i_1^2+i_2^2}}\\
=&-\frac{ici_1\delta_{i_1\,-j_1}\delta_{i_2\,-j_2}+\lr{i_1j_2-i_2j_1}\left[1+\lr{i_1+j_1}^2+\lr{i_2+j_2}^2\right]\phi^{i_1+j_1\,i_2+j_2}}{4\pi^2\lr{1+i_1^2+i_2^2}\lr{1+j_1^2+j_2^2}}=-\mc{J}^{ij}
\end{split}
\end{equation}
The Jacobi identity \eqref{JI_m} reads
\begin{equation}
\begin{split}
\lr{\mc{B}^{im}+\mc{C}^{im}\phi^{i+m}}\mc{C}^{jk}\phi^{j+k}_m+\car=&
\lr{\mc{B}^{i\,j+k}+\mc{C}^{i\,j+k}\phi^{i+j+k}}\mc{C}^{jk}+\car\\=&
\lr{\mc{C}^{i\,j+k}\mc{C}^{jk}+\mc{C}^{j\,k+i}\mc{C}^{ki}+
\mc{C}^{k\,i+j}\mc{C}^{ij}}\phi^{i+j+k}\\&+\mc{B}^{i\,j+k}\mc{C}^{jk}+\mc{B}^{j\,k+i}\mc{C}^{ki}+\mc{B}^{k\,i+j}\mc{C}^{ij}.
\end{split}
\end{equation}
Observe that
\begin{equation}
\begin{split}
\mc{C}^{i\,j+k}\mc{C}^{jk}+\car=&
\frac{\left[i_1\lr{j_2+k_2}-i_2\lr{j_1+k_1}\right]\left[1+\lr{i_1+j_1+k_1}^2+\lr{i_2+j_2+k_2}^2\right]\lr{j_1k_2-j_2k_1}}{16\pi^4\lr{1+i_1^2+i_2^2}\lr{1+j_1^2+j_2^2}\lr{1+k_1^2+k_2^2}}+\car\\
=&\frac{1+\lr{i_1+j_1+k_1}^2+\lr{i_2+j_2+k_2}^2}{16\pi^4\lr{1+i_1^2+i_2^2}\lr{1+j_1^2+j_2^2}\lr{1+k_1^2+k_2^2}}
\{[i_1\lr{j_2+k_2}-i_2\lr{j_1+k_1}]\lr{j_1k_2-j_2k_1}\\
&+[j_1\lr{k_2+i_2}-j_2\lr{k_1+i_1}]\lr{k_1i_2-k_2i_1}+[k_1\lr{i_2+j_2}-k_2\lr{i_1+j_1}]\lr{i_1j_2-i_2j_1}\}=0.
\end{split}
\end{equation}
Similarly,
\begin{equation}
\begin{split}
\mc{B}^{i\,j+k}\mc{C}^{jk}+\car=&ic\frac{i_1\delta_{i_1\,-\lr{j_1+k_1}}\delta_{i_2\,-\lr{j_2+k_2}}\lr{j_1k_2-j_2k_1}}{16\pi^4\lr{1+i_1^2+i_2^2}\lr{1+j_1^2+j_2^2}\lr{1+k_1^2+k_2}}+\car\\=
&\frac{ic}{16\pi^4\lr{1+i_1^2+i_2^2}\left[1+\lr{i_1+k_1}^2+\lr{i_2+k_2}^2\right]\lr{1+k_1^2+k_2^2}}\{-i_1\left[\lr{i_1+k_1}k_2-\lr{i_2+k_2}k_1\right]\\&-\lr{i_1+k_1}\lr{k_1i_2-k_2i_1}-k_1\left[i_1\lr{i_2+k_2}-i_2\lr{i_1+k_1}\right]\}=0.
\end{split}
\end{equation}
We have thus shown that \eqref{JCHM} is a Poisson operator. 
This operator describes the Poisson algebra associated with 
periodic solutions of the Charney-Hasegawa-Mima equation.
Since \eqref{JCHM} is in noncanonical form, the invariant measure 
of the system is not immediately apparent.
We claim that the invariant measure is given by
\begin{equation}
dV=d\phi^0\w d\phi^1\w d\phi^2\w ...\,.\label{IMCHM}
\end{equation}
To see this, define the flow field
\begin{equation}
\dot{\bol{\phi}}=\dot{\phi}^i\p_{i}.
\end{equation}
Here, $\p_{i}$ denotes the $i$th tangent vector in the space of Fourier coefficients. We have
\begin{equation}
\mc{L}_{\dot{\bol{\phi}}}dV=\frac{\p\dot{\phi}^i}{\p\phi^i}\,dV=\p_{i}\lr{\mc{J}^{ij}}H_j\,dV=\sum_{ij}\mc{C}^{ij}\p_i\lr{\phi^{i+j}}H_j=0.
\end{equation}
In the last passage, we used the fact that $\p_i\phi^{i+j}=\delta_{j0}$ and $\mc{C}^{i0}=0$. 

Let $f=f\lr{\phi^0,\phi^1,\phi^2,...}$ denote the distribution function in the space of Fourier coefficients
defined with respect to the invariant measure \eqref{IMCHM}. 
Then, the Fokker-Planck equation for the Charney-Hasegawa-Mima equation can be written 
within the metriplectic formalism as
\begin{equation}
\frac{\p f}{\p t}=\lr{f,E,\Sigma}.
\end{equation}
Here, the metriplectic bracket $\lr{\circ,E,\Sigma}=\pb{\circ}{E}+\db{\circ}{\Sigma}$ is defined in terms of
the Poisson bracket \eqref{PBM} and the dissipative bracket \eqref{DBM} associated with   
the Poisson operator \eqref{JCHM}. Similarly, the total energy $E$ and the constrained entropy $\Sigma$ are
defined according to equation \eqref{obs} in terms of the distribution function $f\lr{\phi^0,\phi^1,\phi^2,...}$, the energy $H\lr{\phi^0,\phi^1,\phi^2,...}$, and the Casimir invariants $C^k\lr{\phi^0,\phi^1,\phi^2,...}$ of the Charney-Hasegawa-Mima equation.
The thermodynamic equilibrium of the system has the form \eqref{feq}. 
Let us evaluate the equilibrium distribution function 
for the case $c=0$ (implying $\lambda=0$).  
The Poisson operator \eqref{JCHM} admits the Casimir invariant
\begin{equation}
C=\frac{1}{2}\int_{\Omega}\left[\phi^2+2\abs{\nabla\phi}^2+\lr{\Delta\phi}^2\right]dxdy=2\pi^2\sum_{n,m=-\infty}^{+\infty}\lr{1+n^2+m^2}^2\abs{\phi^{nm}}^2.
\end{equation}
 Indeed,
\begin{equation}
\mc{J}^{ij}C_j=\sum_{j}\mc{C}^{ij}\phi^{i+j}C_j=
\sum_{j_1,j_2=-\infty}^{+\infty}\frac{\lr{i_1j_2-i_2j_1}\left[1+\lr{i_1+j_1}^2+\lr{i_2+j_2}^2\right]\lr{1+j_1^2+j_2^2}}{1+i_1^2+i_2^2}\phi^{i+j}\phi^{-j}=0.
\end{equation}
Here we used the fact that in the summation 
each term corresponding to a pair $\lr{j_1,j_2}$
cancels with the pair $\lr{j_1',j_2'}=\lr{-i_1-j_1,-i_2-j_2}$. 
Physically, $C$ is the sum of the energy $H$ of \eqref{HCHM} and the enstrophy
$\mc{E}=\frac{1}{2}\int_{\Omega}\left[\abs{\nabla\phi}^2+\lr{\Delta\phi}^2\right]dxdy$.
From \eqref{feq} we thus obtain the equilibrium distribution function
\begin{equation}
\lim_{t\rightarrow\infty}f=\frac{1}{Z}\exp{\left\{
-2\pi^2\left[\sum_{m,n=-\infty}^{+\infty}\lr{1+n^2+m^2}\lr{\beta+\mu\lr{1+n^2+m^2}}\abs{\phi^{nm}}^2\right]
\right\}}.\label{feqCHM}
\end{equation}
Setting $\alpha_{nm}=2\pi^2\lr{1+n^2+m^2}\left[\beta+\mu\lr{1+n^2+m^2}\right]$ and assuming $\beta,\mu\geq 0$, 
the normalization constant (partition function) $Z$ can be evaluated as
\begin{equation}
\begin{split}
Z=&\int\exp{\left\{-\alpha_{nm}\abs{\phi^{nm}}^2\right\}}d\phi^{00}\w d\phi^{01}\w d\phi^{01\ast}\w d\phi^{02}\w d\phi^{02\ast}\w ...\\
=&\sqrt{\frac{\pi}{\alpha_{00}}}\int{\exp{\left\{-2\alpha_{01}\abs{\phi^{01}}^2\right\}}}d\lr{i\theta^{01}}\w d\abs{\phi^{01}}^2\,
\int{\exp{\left\{-2\alpha_{02}\abs{\phi^{02}}^2\right\}}} d\lr{i\theta^{02}}\w d\abs{\phi^{02}}^2\, ...\\
=&\sqrt{\frac{\pi}{\alpha_{00}}}\left[\frac{\pi i}{\alpha_{01}}\int_{0}^{+\infty}e^{-x}dx\right]\left[\frac{\pi i}{\alpha_{02}}\int_{0}^{+\infty}e^{-x}dx\right]...\\
=&\sqrt{\frac{\pi}{\alpha_{00}}}\left[\frac{\pi i}{\alpha_{01}}\right]\left[\frac{\pi i}{\alpha_{02}}\right]...\\
=&\sqrt{\frac{\pi}{\alpha_{00}}}\sqrt{\frac{\alpha_{00}}{\pi i}}\prod_{n,m=-\infty}^{+\infty}\sqrt{\frac{\pi i}{\alpha_{nm}}}\\
=&\frac{1}{\sqrt{i}}\prod_{n,m=-\infty}^{+\infty}\sqrt{\frac{i}{2\pi\lr{1+m^2+n^2}\left[\beta+\mu\lr{1+n^2+m^2}\right]}}.
\end{split}
\end{equation}
Here, we used the fact that $\phi^{nm\ast}=\phi^{-n\,-m}$ (hence, $\phi^{00}$ is real), 
the polar representation $\phi^{nm}=\abs{\phi^{nm}}e^{i\theta^{nm}}$, and 
the property $\alpha_{nm}=\alpha_{-n\,-m}$.

Notice that, if $C$ is a Casimir invariant, so is any function $\nu=\nu\lr{C}$ of $C$.
The question then arises on how to determine the Casimir invariant appearing in
the exponent of the equilibrium distribution function \eqref{feq} without ambiguity.
When solving the Fokker-Planck equation, this function is determined 
automatically by the relaxation process as a consequence of the initial conditions of the system: 
each element of the ensemble preserves the original value of $C$ throughout entropy maximization. 
However, when we deal directly with the equilibrium distribution function 
without prior knowledge of the initial configuration of the system,
the form of the exponent $\mu C$ is postulated (as a function of $C$),  
and the constant $\mu$ is treated as the chemical
potential of a grand canonical ensemble.  

Finally, the basis for the decomposition of the dynamical variables does not need to be
that of the Fourier series. For example, a combination of 
eigenvectors of the curl operator (Beltrami fields) and 
harmonic vector fields can be used to identify 
the invariant measure associated with the evolution of the magnetic field in MHD theories (see \cite{Ito}).
Once the Poisson operator and the invariant measure are known, 
the Fokker-Planck equation and the corresponding equilibrium distribution function can be obtained within the metriplectic formalism
as in the example discussed above.

\section{Metriplectic Manifolds}

The purpose of the present section is to discuss some geometric aspects of dissipative dynamics as described by the metriplectic formalism.
Consider a 2 dimensional canonical Hamiltonian system.
The fundamental geometric structure associated with it is the
symplectic 2 form
\begin{equation}
\omega= dp\w dq=dp\otimes dq-dq\otimes dp,\label{o2}
\end{equation} 
where $\lr{p,q}\in\Omega_c$ are the phase space coordinates
and $\Omega_c$ 2 dimensional phase space. 
Since \eqref{o2} is invertible, 
the Poisson operator $\mc{J}$
corresponds to the inverse of $\omega$:
\begin{equation}
\mc{J}=\p_q \w \p_p=\p_q\otimes\p_p-\p_p\otimes\p_q.\label{J2}
\end{equation}
In matrix representation,
\begin{equation}
\omega=\begin{bmatrix}0&1\\-1&0\end{bmatrix},~~~~\mc{J}=\begin{bmatrix}0&-1\\1&0\end{bmatrix}.
\end{equation} 
Then, the Poisson bracket \eqref{pbm} takes the canonical form
\begin{equation}
\pb{\alpha}{\beta}_{\rm m}=\alpha_q\beta_p-\alpha_p\beta_q,~~~~\alpha,\beta\in C^{\infty}\lr{\Omega_c},
\end{equation}
Let $g$ be the twice contravariant tensor with components
\begin{equation}
g^{ij}=\mc{J}^{ik}\mc{J}^{jk},~~~~i,j=1,...,n,
\end{equation}
appearing in the definition of the dissipative bracket \eqref{dbm}. 
When $\mc{J}$ is given by \eqref{J2}, we have
\begin{equation}
g=\p_p^2+\p_q^2=\p_p\otimes \p_p+\p_q\otimes\p_q,
\end{equation}
or, in matrix form, 
\begin{equation}
g=-\mc{J}^2=I=\begin{bmatrix}1&0\\0&1\end{bmatrix},
\end{equation}
where $I$ is the 2 dimensional identity matrix.
Thus, the `canonical' form for the dissipative bracket \eqref{dbm} is
\begin{equation}
\db{\alpha}{\beta}_{\rm m}=\alpha_p\beta_p+\alpha_q\beta_q,~~~~\alpha,\beta\in C^{\infty}\lr{\Omega_c}. 
\end{equation}
The inverse of the tensor $g$  
defines a twice covariant non-degenerate symmetric tensor  
\begin{equation}
\mf{g}=dp^2+dq^2=dp\otimes dp+dq\otimes dq.\label{symmt}
\end{equation}
Evidently, $\mf{g}$ is a metric tensor on $\Omega_c$. 
This calculation suggests that the essential 
geometric structure associated with dissipative dynamics
as described by \eqref{dbm} is the `canonical metric tensor' \eqref{symmt}.
Recalling \eqref{Xdotbm}, the change $\dot{H}_d$ in energy $H\lr{p,q}$ due
to dissipation can be evaluated as
\begin{equation}
\dot{H}_d=-\db{H}{\gamma H}_{\rm m}=-\gamma \lr{H_p^2+H_q^2}=-\gamma\lr{\dot{p}^2+\dot{q}^2}.
\end{equation}
Hence, energy dissipation is proportional to the square of (unperturbed) phase space velocity. 

When the hypothesis of the Lie-Darboux theorem discussed in section 2
are verified, a similar result applies in dimensions greater than 2.
In particular, it is possible to locally express the tensor $g$ as
\begin{equation}
g=\sum_{i=1}^{r}\p_{p^i}\otimes\p_{p^i}+\p_{q^i}\otimes\p_{q^i}.
\end{equation}
and the tensor $\mf{g}$ as  
\begin{equation}
\mf{g}=\sum_{i=1}^{r}dp^{i}\otimes dp^i+dq^i\otimes dq^i.
\end{equation}
Notice that, however, both $g$ and $\mf{g}$ are degenerate in general,
and $\mf{g}$ defines a metric tensor only over the submanifold 
$\Omega_C=\left\{\bol{x}\in\Omega~\rvert~ C^1=c^1,...,C^m=c^m\right\}$ with $c^1,...,c^m\in\mathbb{R}$. 

We conclude this section with a list of the main geometric constructions
that occur in the algebraic formulation of dissipative dynamics
discussed in the present paper. 
We restrict our attention to finite dimensions.

Let $\Omega\subset\mathbb{R}^n$ denote a smoothly bounded domain 
with boundary $\p\Omega$. 

\begin{mydef} 
An antisymmetric bracket on $\Omega$ is a bilinear map 
over the field of real numbers
\begin{equation}
\left\{\cdot,\cdot\right\}:C^{\infty}\lr{\Omega}\times C^{\infty}\lr{\Omega}\rightarrow C^{\infty}\lr{\Omega},
\end{equation}
such that
\begin{subequations}
\begin{align}
&\left\{f,f\right\}=0,\\
&\left\{f,g\right\}=-\left\{g,f\right\},\\
&\left\{fg,h\right\}=f\left\{g,h\right\}+\left\{f,h\right\}g,
\end{align}
\end{subequations}
for all $f,g,h\in C^{\infty}\lr{\Omega}$.
\end{mydef}
The antisymmetric bracket is the structure required for conservation of energy.

\begin{mydef}
A Liouville or measure preserving bracket on $\Omega$ is an antisymmetric bracket $\left\{\cdot,\cdot\right\}$ on $\Omega$ such that
\begin{equation}
\left\{f,g\right\}={\rm div}\lr{f X_g},\label{Lb}
\end{equation} 
for all $f,g\in C^{\infty}\lr{\Omega}$. Here, $X_g=\left\{x^i,g\right\}\p_i$ is the flow generated by $g$, $\lr{x^1,...,x^n}$ a coordinate system in $\Omega$, and the divergence is calculated
with respect to the volume element $JdV=Jdx^1\w ... \w dx^n$.
\end{mydef}

The Liouville bracket is the structure required to formulate
the Fokker-Planck equation \eqref{FPE1}. 
Indeed, it assigns an invariant measure, which is needed for
the formulation of statistical mechanics. 
To see this, suppose that the Liouville bracket has the form
$\left\{f,g\right\}=f_i\mc{J}^{ij}g_j$ for some bivector field $\mc{J}$. 
Then, in the coordinate system $\lr{x^1,...,x^n}$
with volume element $JdV$ equation \eqref{Lb} reads as  
\begin{equation}
\left\{f,g\right\}={\rm div}\lr{fX_g}=\frac{1}{J}\frac{\p}{\p x^i}\lr{Jf\mc{J}^{ij}g_j}=
\left\{f,g\right\}+\frac{f}{J}\p_i\lr{J\mc{J}^{ij}}g_j~~~~\forall f,g.
\end{equation}
It follows that
\begin{equation}
\p_i\lr{J\mc{J}^{ij}}=0,~~~~j=1,...,n.
\end{equation}
This is exactly the condition \eqref{IM} for
the existence of an invariant measure $JdV$ 
regardless of the choice of the Hamiltonian function. 

\begin{mydef}
A Poisson bracket on $\Omega$ is an antisymmetric bracket $\left\{\cdot,\cdot\right\}$ on $\Omega$ such that
\begin{equation}
\left\{f,\left\{g,h\right\}\right\}+\car=0,
\end{equation}
for all $f,g,h\in C^{\infty}\lr{\Omega}$.
\end{mydef}

The Poisson bracket characterizes the algebraic structure of Hamiltonian systems. 
Due to the Lie-Darboux and Liouville's theorems,  
a Poisson bracket $\left\{f,g\right\}=f_i\mc{J}^{ij}g_j$ 
is locally a Liouville bracket in regions where the rank of $\mc{J}$ is constant.

\begin{mydef}
A dissipative bracket on $\Omega$ is a bilinear map over the field of real numbers
\begin{equation}
\left[\cdot,\cdot\right]:C^{\infty}\lr{\Omega}\cp C^{\infty}\lr{\Omega}\rightarrow C^{\infty}\lr{\Omega},
\end{equation}
such that 
\begin{subequations}
\begin{align}
&\db{f}{f}\geq 0,\\
&\db{f}{g}=\db{g}{f},\\
&\db{fg}{h}=f\db{g}{h}+\db{f}{h}g,\\
\end{align}
\end{subequations}
for all $f,g,h\in C^{\infty}\lr{\Omega}$.
\end{mydef}

The metriplectic bracket is thus obtained by combining a Poisson bracket with a dissipative bracket.
\begin{mydef}
A metriplectic bracket on $\Omega$ is a map
\begin{equation}
\lr{\cdot,\cdot,\cdot}:C^{\infty}\lr{\Omega}\cp C^{\infty}\lr{\Omega}\cp C^{\infty}\lr{\Omega}\rightarrow C^{\infty}\lr{\Omega},
\end{equation}
such that 
\begin{equation}
\lr{f,g,h}=\pb{f}{g}+\db{f}{h},
\end{equation}
for all $f,g,h\in C^{\infty}\lr{\Omega}$ and where $\pb{\cdot}{\cdot}$ is a Poisson bracket on $\Omega$
and $\db{\cdot}{\cdot}$ a dissipative bracket on $\Omega$.
\end{mydef} 
Observe that a Poisson bracket $\pb{f}{g}=f_i\mc{J}^{ij}g_j$ naturally induces a dissipative bracket
$\db{f}{g}=f_i\mc{J}^{ik}\mc{J}^{jk}g_j$ and an associated metriplectic bracket 
$\lr{f,g,h}=f_i\mc{J}^{ij}g_j+f_i\mc{J}^{ik}\mc{J}^{jk}h_k$.  
In other words, a Poisson manifold $\lr{\Omega,\mc{J}}$ is also a metriplectic manifold $\lr{\Omega,\mc{J},g}$ with $g=\mc{J}^{ik}\mc{J}^{jk}\p_i\otimes \p_j$. We refer to a metriplectic manifold induced by a Poisson manifold as a natural
metriplectic manifold.


\section{Concluding Remarks}

In this paper, we constructed the metriplectic bracket
that generates the Fokker-Planck equation associated with
diffusion processes in Hamiltonian systems.
Such macroscopic bracket is induced from the microscopic
stochastic dynamics of the particles that populate the statistical ensemble.
We found that, on the invariant measure assigned by
Liouville's theorem, the dissipative bracket is fully determined
by the microscopic Poisson operator, and exhibits a double bracket form.
This result suggests a canonical form of the dissipative bracket,
which is characterized by an Euclidean metric tensor  
on the Casimir leaves spanned by the canonical coordinates provided
by the Lie-Darboux theorem.
As an application to the statistics of infinite dimensional Hamiltonian systems, 
we discussed the Fokker-Planck formalism for the Charney-Hasegawa-Mima equation.

We remark that, while the theory discussed in this study
applies to general Hamiltonian systems, 
the resulting Fokker-Planck equation does not necessarily
correspond to standard equations that incorporate dissipative effects.  
For example, if one implements the present contruction to the ideal Euler equations, 
the Fokker-Planck equation will introduce dissipation while
respecting the relevant topological invariants of the fluid flow such as helicity. This setting is clearly different from that of the Navier-Stokes equations with finite viscosity. 
Therefore, the applicability of the theory developed here is contingent upon careful considerations on the kind of relaxation process that is physically relevant for the system of interest.

\section*{Acknowledgments}

\noindent The research of N. S. was supported by JSPS KAKENHI Grant No. 18J01729.
          N. S. is grateful to Professor P. J. Morrison for useful discussion   
					on noncanonical Hamiltonian systems and the metriplectic formalism of dissipative dynamics, to Professor M. Yamada for useful discussion on the Charney-Hasegawa-Mima equation, and to Professor Z. Yoshida for useful criticism on the statistical mechanics of constrained systems.  
					
\appendix

\section{Lie-Darboux Theorem}

\begin{theorem}
Let $\omega$ denote a closed smooth $2$ form of rank $2r=n-m$ in a domain $\Omega\subset\mathbb{R}^{n}$. Then, for every point $\bol{x}\in\Omega$
there exist a neighborhood $U$ of $\bol{x}$ and local
coordinates $\lr{p^1,...,p^r,q^1,...,q^r,C^1,...,C^m}$ such that
\begin{equation}
\omega=\sum_{i=1}^{r}dp^i\w dq^i~~~~{\rm in}~~U.\label{LD}
\end{equation}
\end{theorem}

\begin{proof}
Let $\lr{x^1,...,x^n}$ denote a Cartesian coordinate system in $\Omega$.
The kernel of $\omega$ is spanned by $m$ smooth orthonormal tangent vectors 
$\bol{\xi}_i=\xi_i^{j}\p_j\in T\Omega$, $i=1,...,m$. 
Let $\lr{\bol{\theta}_1,...,\bol{\theta}_{2r},\bol{\xi}_1,...,\bol{\xi}_m}$  
denote an orthonormal basis of smooth tangent vectors in $\Omega$.  
To each $\bol{\xi}_i$ and $\bol{\theta}_i$ we assign the cotangent vectors $\xi_i=\sum_{j=1}^{n}\xi_i^{j} dx^j\in T^\ast\Omega$ and $\theta_i=\sum_{j=1}^n\theta_{i}^jdx^j\in T^\ast\Omega$. Then, $\lr{\theta_1,...,\theta_{2r},\xi_1,...,\xi_m}$ forms a smooth basis of the cotangent bundle $T^\ast\Omega$ such that
\begin{equation}
\omega=\sum_{i<j}\alpha_{ij}\theta_i\w\theta_j.
\end{equation}
The 2 form $\omega$ is closed. This implies
\begin{equation}
\sum_{i<j}d\alpha_{ij}\w\theta_i\w\theta_j+\sum_{ij}\alpha_{ij}d\theta_i\w\theta_j=0.
\end{equation}
Multiplying this expression by the $2r-1$ form $\theta_k^{2r-1}=\theta_1\w ... \w\theta_{k-1}\w\theta_{k+1}\w...\w\theta_{2r}$ we obtain
\begin{equation}
\sum_{i}\alpha_{ik}\theta_1\w ...\w\theta_{k-1}\w\theta_k\w\theta_{k+1}\w ...\w\theta_{2r}\w d\theta_{i}=0. \label{Frob0}
\end{equation} 
Since by hypothesis $\omega$ has rank $2r$, the matrix $\alpha_{ik}$ is invertible with inverse $\lr{\alpha^{-1}}^{kj}$. 
Multiplying the left-hand side of \eqref{Frob0} by $\lr{\alpha^{-1}}^{kj}$ and summing over $k$ gives
\begin{equation}
\theta_1\w ...\w\theta_{2r}\w d\theta_j=0,~~~~j=1,...,2r.\label{Frob1}
\end{equation}
System \eqref{Frob1} is the Frobenius integrability condition \cite{Frankel} for
the kernel of $\omega$, 
\begin{equation}
{\rm ker}\lr{\omega}=\left\{X\in T\Omega~\rvert~i_X\theta_i=0~\forall i=1,...,2r\right\}.
\end{equation}
Hence, for each $\bol{x}\in\Omega$, there exists a neighborhood  
$V\subset\Omega$ of $\bol{x}$ and local coordinates
$\lr{y^1,...,y^{2r},C^1,...,C^m}$ such that the submanifolds $y^1={\rm const.}, ...,y^{2r}={\rm const.}$ are integral manifolds of ${\rm ker}\lr{\omega}$
and the cotangent vectors $\theta_i$ take the form $\theta_i=\tau_{ij}dy^j$  for some smooth coefficients $\tau_{ij}=\tau_{ij}\lr{y^1,...,y^{2r},C^1,...,C^m}$, $i,j=1,...,2r$. It follows that 
\begin{equation}
\omega=\sum_{i<j}A_{ij}dy^i\w d y^j~~~~{\rm in}~~V,\label{omegay}
\end{equation}
for some smooth coefficients $A_{ij}=A_{ij}\lr{y^1,...,y^{2r},C^1,...,C^m}$, $i,j=1,...,2r$. 
Furthermore, using the condition $d\omega=0$ with equation \eqref{omegay}, one obtains $\p A_{ij}/\p C^{k}=0$, $k=1,..,m$,  
which implies $A_{ij}=A_{ij}\lr{y^1,...,y^{2r}}$, $i,j=1,...,2r$.  
Thus, the problem is now reduced to the standard non-degenerate case 
on the $2r$ dimensional submanifold $V_C=\left\{\bol{x}\in\Omega~\rvert~C^1={\rm const.},...,C^{m}={\rm const.}\right\}$
where the $2$ form $\omega$ has full rank $2r$. The proof can be obtained accordingly (see \cite{Arnold,deLeon}). 
\end{proof}

\end{document}